\documentclass[runningheads]{llncs}

\usepackage[T1]{fontenc}
\usepackage{graphicx}
\usepackage{bm}
\usepackage{array}
\usepackage{adjustbox}
\usepackage{hyperref}
\hypersetup{hidelinks}

\usepackage{comment}
\usepackage{amsmath}
\usepackage{booktabs}
\usepackage[capitalise]{cleveref}
\crefname{equation}{condition}{conditions}

\usepackage{tikz}
\usetikzlibrary{calc}
\usetikzlibrary {arrows.meta}
\usetikzlibrary{automata, positioning, arrows}
\usepackage[all,cmtip,2cell]{xy}
\usepackage{tikz-cd}
\usepackage{wrapfig}
\usepackage{caption}
\usepackage{subcaption}
\usepackage{nicefrac}
\usepackage{mdframed}

\usepackage{framed}

\usepackage{fancybox}

\usepackage{amsfonts}

\usepackage{multirow}
\usepackage{threeparttable}

\usepackage{stmaryrd}

\usepackage{svg}
\svgpath{{imgs/}}
\svgsetup{
	inkscapelatex=false
}

\usepackage{listings}
\lstdefinelanguage{prog}
{
	morekeywords={if, then, else, fi, while, do, od, true, false, and, or, skip, sample, observe,return,score,normal,uniform, prob},
	sensitive = false
}

\newif\ifdraft
\drafttrue

\newcommand{\inlsec}[1]{\smallskip\noindent\textbf{#1.}}
\newcommand{\inlsecit}[1]{\smallskip\noindent\textit{#1.}}

\newcommand{\probm}{\mathbb{P}}

\newcommand{\Rset}{\mathbb{R}}
\newcommand{\Nset}{\mathbb{N}}

\newcommand{\expv}{\mathbb{E}}

\newcommand{\pv}{\mathbf{v}}
\newcommand{\rv}{\mathbf{r}}

\newcommand{\lin}{\loc_\mathrm{init}}
\newcommand{\lout}{\loc_\mathrm{out}}

\newcommand{\loc}{\ell}

\newcommand{\valin}{\pv_{\mathrm{init}}}

\newcommand\calV{\mathcal{V}}
\newcommand{\tposterior}{\mathsf{TPD}}



\begin{document}

\pagestyle{plain}

\title{Quantitative Verification of Omega-regular Properties in Probabilistic Programming}

\titlerunning{Quantitative Verification of Omega-regular Properties}

\author{Peixin Wang\inst{1} \and
Jianhao Bai\inst{1} \and
Min Zhang\inst{1} \and
C.-H. Luke Ong\inst{2}}

\authorrunning{Wang et al.}

\institute{East China Normal University, China \and
Nanyang Technological University, Singapore}

\maketitle            

\makeatletter
\renewcommand{\@author}{}
\makeatother


 \begin{abstract}

Probabilistic programming provides a high-level framework for specifying statistical models as executable programs with built-in randomness and conditioning. Existing inference techniques, however, typically compute posterior distributions over program states at fixed time points—most often at termination—thereby failing to capture the temporal evolution of probabilistic behaviors. We introduce \emph{temporal posterior inference} (TPI), a new framework that unifies probabilistic programming with temporal logic by computing posterior distributions over execution traces that satisfy $\omega$-regular specifications, conditioned on (possibly temporal) observations. To obtain rigorous quantitative guarantees, we develop a new method for computing upper and lower bounds on the satisfaction probabilities of $\omega$-regular properties. Our approach decomposes Rabin acceptance conditions into persistence and recurrence components and constructs stochastic barrier certificates that bound each component soundly. We implement our approach in a prototype tool, \textsc{TPInfer}, and evaluate it on a suite of benchmarks, demonstrating effective and efficient inference over rich temporal properties in probabilistic models.

\end{abstract}
 
 \section{Introduction}\label{sec:intro}
 
Probabilistic programming (PP) offers a high-level approach to describing stochastic models by embedding random variables and probabilistic choices directly into program code \cite{gordon2014probabilistic,goodman2012church}. Given observations, PP systems automatically apply Bayes' rule to compute posterior distributions \cite{jaynes2003probability,bishop2006pattern}, providing a principled mechanism for reasoning about latent uncertainty. This makes PP particularly well suited for modern AI systems, which must operate under noisy data, stochastic dynamics, and incomplete information \cite{ghahramani2015probabilistic}.

However, standard posterior inference (SPI) in PP typically concentrates on distributions over program states at fixed time points---most notably at program termination~\cite{van2018introduction}. This implicitly assumes almost-sure termination (AST)~\cite{chakarov2013probabilistic} of probabilistic programs that may not hold~\cite{feng2023lower} and fails to capture the temporal evolution of probabilistic behaviors that unfolds along program execution. To illustrate this limitation, we present a simple example below.

\begin{example}\label{ex:running-example}
Consider a simple probabilistic program in~\cref{fig:running-example}, where $x$ stochastically toggles its value between $0,1$ and $y$ counts the steps in which $x=1$. The posterior probability $\Pr(x=1 \mid y>0)$ denotes the probability that $x$ equals~$1$ at termination, conditioned on observing $y>0$.
\end{example}

\begin{figure}[htbp]
    \begin{minipage}{0.45\textwidth}
    \vspace{1cm}
    \subfloat[A simple probabilistic program]{
        \hspace{1cm}
        \begin{minipage}{0.9\textwidth}
            \raggedright
            \texttt{$x \sim \mathrm{Bernoulli}(0.3)$; $y := 0$; \\
            $\bf{while}$ $\mathrm{Bernoulli}(0.5)$ $\bf{do}$ \\
            \qquad $\bf{if}$ $\mathbf{prob}(0.6)$ $\bf{then}$ \\
            \qquad \qquad $x := 1 - x$ \\
            \qquad $\bf{if}$ $x = 1$ $\bf{then}$ \\
            \qquad \qquad $y := y + 1$ \\
            $\bf{done}$}
        \end{minipage}
    }
\end{minipage}
    \hspace{6\textwidth}

    \begin{minipage}{1.5\textwidth}
    \vspace{-5.5cm}
        \centering
        \subfloat[F$(x=1)\ \& \ $F$(y>0)$]{
            \hspace{-1cm}
            \scalebox{0.5}{
                \hspace{-1cm}
                \begin{tikzpicture}[minimum size=10mm, node distance=30mm, every node/.style={font=\small}]
                    \node[draw, circle, fill=none] (3) {\textbf{0}};
                    \node[draw, circle, fill=none, right of=3] (1) {\textbf{1}};
                    \node[draw, circle, fill=none, below of=1, node distance=20mm] (2) {\textbf{2}};
                    \node[draw, circle, fill=none, right of=1] (0) {\textcolor{red}{\textbf{3}}};
                    \node[fill=none, left of=3] (empty) {};
                    
                    \draw[->, shorten <= 10mm] (empty) -- (3);
                    \draw[->] (3) edge[loop above] node[above] {$x\neq 1 \  \& \ y\leq0$} (3)
                              edge[->] node[left] {$x= 1 \  \& \ y\leq0$} (2)
                              edge[->] node[above] {$x\neq 1 \  \& \ y>0$} (1)
                              edge[->, bend left=50] node[above] {$x=1 \  \& \ y>0$} (0);
                    
                    \draw[->] (0) edge[loop above] node[above] {$1$} (0);
                    
                    \draw[->] (1) edge[loop below] node[right] {$x\neq1$} (1)
                              edge[->] node[above] {$x=1$} (0);
                    
                    \draw[->] (2) edge[loop below] node[right] {$y\leq0$} (2)
                              edge[->] node[right] {$y > 0$} (0);
                \end{tikzpicture}
            }
        }

        \vspace{0cm} 

        \subfloat[G$(x=1)\ \& \ $F$(y>0)$]{
            \hspace{-0.3cm}
            \scalebox{0.7}{
                \begin{tikzpicture}[minimum size=7mm, node distance=10mm, every node/.style={font=\small}]
                    \node[draw, circle, fill=none] (1) {\textcolor{blue}{\textbf{0}}}; 
                    \node[draw, circle, fill=none, right of=0,xshift=-20mm] (2) {\textcolor{blue}{\textbf{2}}};
                    
                    \node[fill=none, right of=1, node distance=25mm] (empty) {};
                    \node[draw, circle, fill=none, above of=empty] (0) {\textcolor{red}{\textbf{1}}}; 
                    \node[fill=none, left of=1] (empty) {};

                    \draw[->] (empty) -- (1);
                    
                    \draw[->] (0) edge[loop above] node[above] {$x=1$} (0)
                              edge[->] node[above] {$x\neq 1$} (2);

                    \draw[->] (1) edge[loop below] node[below] {$x=0\ \& \ y\leq0$} (1)
                               edge[->] node[below] {$x\neq 1$} (2)
                               edge[->] node[left] {$x=1\ \& \ y>0$} (0);

                    \draw[->] (2) edge[loop above] node[above] {$1$} (2);

                \end{tikzpicture}
            }
        }
    \end{minipage}
    \caption{Running-example}
    \label{fig:running-example}
\end{figure}
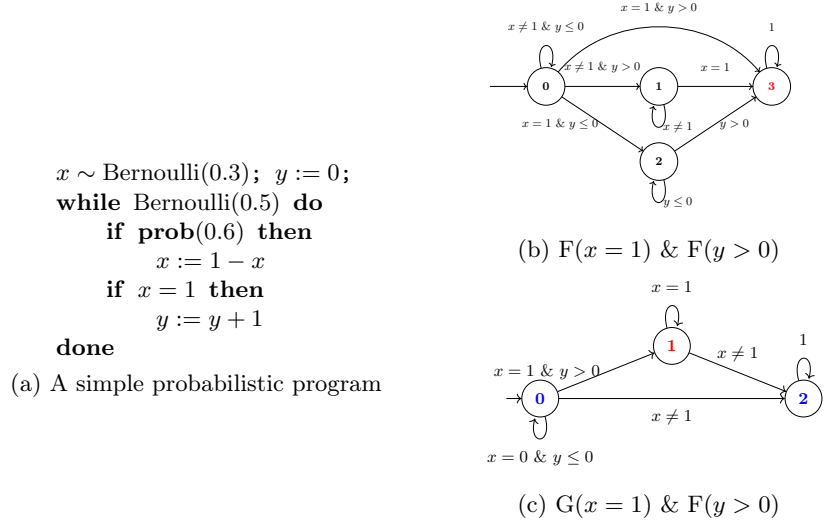

In \cref{ex:running-example}, this limitation becomes apparent: we are not merely interested in the final value of $x$, but in the probability that $x$ remains equal to $1$ throughout the entire execution, conditioned on finally observing $y>0$. This temporal requirement is captured in Linear Temporal Logic (LTL)~\cite{pnueli1977temporal} by the ``always'' operator $\mathbf{G}$ and ``eventually'' operator $\mathbf{F}$, yielding the posterior probability $\Pr(\mathbf{G}\, x = 1 \mid \mathbf{F} y>0)$.\footnote{Here we assume a self-loop at program termination so that we have infinite traces that can be specified by LTL.} 
In fact, the posterior probability $\Pr(x=1 \mid y>0)$ is equivalent to $\Pr(\mathbf{F}\, x = 1 \mid \mathbf{F} y>0)$ under LTL, and thus produces different automata.
Furthermore, many important properties in probabilistic systems, such as safety (``something bad never happens'') and liveness (``something good eventually happens''), depend on how program states evolve over time rather than on a single terminating state, making temporal reasoning indispensable~\cite{baier2008principles}. 

\inlsec{Motivation}
This limitation reveals a gap between SPI and temporal reasoning. SPI is variable-centric and focused on terminal states, whereas temporal reasoning is trace-centric and concerns whole executions. Neither alone can express conditional reasoning about temporal behaviors or evidence that unfolds over time. To bridge this gap, we introduce \emph{temporal posterior inference} (TPI), which computes posterior probabilities over execution traces satisfying $\omega$-regular specifications~\cite{thomas1990automata}, conditioned on (possibly temporal) observations. This unifies probabilistic inference with temporal reasoning and enables likelihood-based analysis of temporal properties.
\begin{framed}
\noindent \textbf{TPI}. Given $\omega$-regular properties $\varphi$ and $\psi$, \textbf{the objective is to compute $\Pr(\varphi \mid \psi)$}, the probability that traces satisfy $\varphi$ conditioned on satisfying $\psi$. 	
\end{framed}

\inlsec{Challenges}
TPI presents two main challenges.
(i) It requires reasoning over entire probabilistic execution traces rather than isolated states; such traces may be unbounded, branching, and exhibit long-range temporal dependencies, making their distributions difficult to analyze. Previous work mainly focuses on termination~\cite{chakarov2013probabilistic,chatterjee2016algorithmic,DBLP:conf/vmcai/FuC19},  or non-temporal properties like sensitivity~\cite{DBLP:journals/pacmpl/BartheEGHS18,DBLP:journals/pacmpl/WangFCDX20}, expectation~\cite{DBLP:conf/pldi/NgoC018,DBLP:conf/pldi/Wang0GCQS19,DBLP:conf/tacas/BatzCJKKM23}, etc.
(ii) For stochastic models, computing quantitative bounds on the satisfaction of $\omega$-regular properties is inherently complex: these properties encode infinite-horizon persistence and recurrence conditions, and conditioning on temporal observations introduces global constraints that drastically complicate the structure of admissible traces. \cite{chakarov2016deductive} cares about qualitative persistence and recurrence. While \cite{abate2025quantitative,henzinger2025supermartingale} work on quantitative temporal verification, they do not consider posterior inference for temporal reasoning. A recent work~\cite{watanabe2025unifying} proposes a coalgebraic framework for quantitative temporal inference but it does not support automated calculation. 

\inlsec{Our approaches}
To address these challenges, we adopt an automata-theoretic approach that enables us to track the potentially infinite temporal evolution of probabilistic programs. In particular, we translate temporal specifications into deterministic Rabin automata (DRAs)~\cite{thomas1990automata}, which provide a structured way to reason about infinite execution traces and to identify the temporal patterns relevant for verification. Building on this foundation, we aim to compute quantitative bounds on the satisfaction probabilities of $\omega$-regular properties under observational conditioning. To make this tractable, we decompose the Rabin acceptance condition into two fundamental components---\emph{persistence} and \emph{recurrence}---corresponding to events that must eventually stabilize or occur infinitely often. For each component, we develop new stochastic barrier certificates~\cite{prajna2004safety} that yield sound upper and lower bounds on the probabilities of satisfying these temporal patterns. By combining the bounds for persistence and recurrence, we obtain quantitative guarantees for the full $\omega$-regular specification, thereby enabling rigorous probabilistic reasoning about temporal behavior in systems with stochastic dynamics.

\inlsec{Contributions}
This paper makes the following contributions:
\begin{enumerate}
	\item We introduce a new framework,  \emph{temporal posterior inference}, which unifies probabilistic inference with temporal-logical reasoning by computing posterior distributions over execution traces that satisfy $\omega$-regular properties.
	
	\item We build an automata-theoretic foundation for TPI by decomposing the Rabin acceptance condition into persistence and recurrence components, enabling systematic tracking of the potentially infinite temporal evolution of probabilistic programs.
	
	\item We provide quantitative analysis methods for the persistence and recurrence components, and devise corresponding new stochastic barrier certificates that yield sound upper and lower bounds on their satisfaction probabilities.
	
	\item We implement our approach into a prototype tool  \textsc{TPInfer} and evaluate it on a suite of benchmarks, demonstrating its practicality and effectiveness for quantitative reasoning about temporal properties in stochastic systems.
\end{enumerate}

All omitted proofs are put in the Appendix.

 \section{Preliminaries and Probabilistic Programming}\label{sec:pre}
\subsection{LTL and Omega-regular Properties}\label{sec:omega-prelim}
Linear Temporal Logic (LTL)~\cite{pnueli1977temporal} specifies how system behaviors evolve along linear-time executions using temporal operators.

\inlsecit{Syntax of LTL}
Let $\mathit{AP}$ be a finite set of atomic propositions.  
The formulas of LTL over infinite traces are defined by the following grammar:
\[
\varphi ::= \top \mid  a \mid \neg\varphi \mid \varphi_1 \wedge \varphi_2 
           \mid \mathbf{X}\varphi \mid \varphi_1\ \mathbf{U}\ \varphi_2
\]
where $a \in \mathit{AP}$ is an atomic proposition, $\mathbf{X}$ is the ``next'' operator, and $\mathbf{U}$ is the ``until'' operator.  
Other standard temporal operators are defined as syntactic sugar, e.g., 
$\mathbf{F}\varphi \equiv \top\ \mathbf{U}\ \varphi$ (eventually), 
$\mathbf{G}\varphi \equiv \neg\mathbf{F}\neg\varphi$ (always).

\inlsecit{Semantics of LTL}
Let $Words(\varphi)=\{\sigma=(A_0,A_1,\dots)\in (2^{\mathit{AP}})^\omega \mid \sigma\models \varphi \}$ be the LTL property induced by $\varphi$. Denote by $\sigma^{i+}=(A_i,A_{i+1},\dots)$ a word starting from index $i\ge 0$.  Then the satisfaction relation $\sigma\models \varphi$ is defined inductively:

\[
\begin{array}{lcr}
	\sigma \models \top 
	\qquad\qquad\qquad\quad
	\sigma \models a \text{ iff } a \in A_0 
	\qquad\qquad\qquad
	\sigma \models \neg\varphi \text{ iff } \sigma \not\models \varphi
	\\
	
	\sigma \models \varphi_1 \wedge \varphi_2 
	 \text{ iff } 
	\sigma \models \varphi_1 \ \text{and}\ \sigma \models \varphi_2 
	\qquad\qquad\qquad
	\sigma \models \mathbf{X}\varphi \text{ iff } \sigma^{1+} \models \varphi
	\\
	
	\sigma \models \varphi_1\,\mathbf{U}\,\varphi_2 
	\text{ iff }
	\exists j \ge 0.\ \sigma^{j+} \models \varphi_2 
	\ \text{and}\ 
	\forall\,0 \le i < j.\ \sigma^{i+} \models \varphi_1
	\\
	
	\sigma \models \mathbf{F}\varphi \text{ iff } \exists j \ge 0.\ \sigma^{j+} \models \varphi
	\qquad\qquad\qquad
	\sigma \models \mathbf{G}\varphi \text{ iff } \forall j \ge 0.\ \sigma^{j+} \models \varphi
\end{array}
\]

$\omega$-regular properties~\cite{thomas1990automata} form a rich class that subsumes LTL. 
They are recognized by $\omega$-automata, which accept infinite words using specialized acceptance conditions. Below we introduce the $\omega$-automata used in our work.

\inlsec{DRAs}
A \emph{deterministic Rabin Automaton} (DRA) is a tuple $\mathcal{A}=(Q,q_0,\Sigma,\delta,\mathit{Acc})$ where $Q$ is a finite set of states, $q_0\in Q$ is the initial state, $\Sigma$ is a finite alphabet, $\delta:Q\times \Sigma\rightarrow Q$ is the transition relation, and $\mathit{Acc}=\{(E_i,F_i)\mid E_i,F_i\subseteq Q, E_i \cap F_i = \emptyset, 1\le i\le M\}$ is the acceptance condition.

\inlsecit{Semantics of DRAs}
A  \emph{word} of $\mathcal{A}$, $\sigma=(\sigma_0,\sigma_1,\dots)\in \Sigma^\omega$, is an infinite sequence of letters. 
A \emph{run} on a word $\sigma$, $\rho=(q_0,q_1,\dots)\in Q^\omega$, is an infinite sequence of states such that $q_{i+1}=\delta(q_i,\sigma_i)$ for all $i\in \Nset_0$. Denote by $\Xi$ the set of all runs by $\mathcal{A}$.
Let $\mathrm{Inf}(\rho)$ be the set of states in $\rho$ that are visited infinitely often. 
A run $\rho$ on a word $\sigma$ is \emph{accepting} iff for some $i\in [1,M]$, $\mathrm{Inf}(\rho)\cap E_i=\emptyset$ and $\mathrm{Inf}(\rho)\cap F_i\neq\emptyset$, then the word  $\sigma$ is  \emph{accepted} by $\mathcal{A}$.
The \emph{language} of $\mathcal{A}$,  denoted by $\mathcal{L}(\mathcal{A})$, is the set of all words accepted by $\mathcal{A}$. 
\subsection{Probabilistic Programming}\label{sec:PPL}
We assume that the readers are familiar with notions from probability theory such as probability measure, random variable and expected value. 

\inlsecit{Syntax of PPLs}
We consider standard imperative probabilistic programming language (PPL) that includes the usual loop, conditional and sequential structures. The syntax is given as follows:
\begin{align*}
	S &::= \textbf{skip} \mid x:=E\mid x\sim D   \mid \mbox{\textbf{while}}\, B \, \text{\textbf{do}} \  S \, \text{\textbf{done}}\\
	&  \mid \mbox{\textbf{if}} \, \mbox{$B$}\,\mbox{\textbf{then}} \,  S_1 \, \mbox{\textbf{else}} \,S_2 \,\mbox{\textbf{fi}}\mid \mbox{\textbf{if}} \, \mbox{\textbf{prob}($p$)}\,\mbox{\textbf{then}} \,  S_1 \, \mbox{\textbf{else}} \,S_2 \,\mbox{\textbf{fi}}  \mid S_1;S_2 \\
	B&::=\textbf{true} \mid\neg B\mid B_1\, \textbf{and} \, B_2 \mid   E_1\le E_2 \\
	E&::= x\mid c\mid E_1+E_2\mid E_1 \times E_2  \ \ \ 
	D ::=   \textbf{uniform}(a,b)\mid \cdots 
\end{align*}
where $a,b \in\Rset$ are constants, $p\in [0,1]$,  and the metavariables $S$, $B$, $E$ and $D$ stand for statements, conditions, arithmetic expressions, and distributions (that are discrete or continuous), respectively.  The statement ``$x\sim D$'' is a sampling assignment that draws from the distribution $D$ to obtain a sample value $r$ and then assigns it to $x$. The statement ``$\mbox{\textbf{if}} \, \mbox{\textbf{prob}($p$)}\,\mbox{\textbf{then}} \,  S_1 \, \mbox{\textbf{else}} \,S_2 \,\mbox{\textbf{fi}}$ '' is a probabilistic-branching statement that executes $S_1$ with probability $p$, or $S_2$ with probability $1-p$.

Let $V=\{x_1,\dots,x_n\}$ be a set of real-valued program variables and $R=\{r_1,\dots,r_m\}$ be a set of real-valued sample variables.

\inlsecit{Semantics of PPLs} 
Given a probabilistic program $P$, it can be modeled as a \emph{probabilistic transition system} (PTS) $\Pi=(L,V,R,\mathcal{D}_R,\mathcal{T},\lin,\lout,\mathcal{D}_0)$ where: 
\begin{itemize}
	\item $L$ is a finite set of program locations, $\lin\in L$ is the initial location and $\lout\in L$ is the terminal location;
	\item $V,R$ are finite sets of program and sample variables;
\item $\mathcal{D}_R$ is a function that assign each sample variable $r\in R$ a probability distribution $\mathcal{D}(r)$, and $\mathcal{D}_R$ can be treated as the joint distribution of all sample variables in $R$;
	\item $\mathcal{T}$ is a finite set of transitions. Each transition $\tau\in\mathcal{T}$ is a tuple $(\loc,\phi,f_1,\dots,f_u)$ such that (i) $\loc$ is the source location, (ii) $\phi$ is the guard predicate over $V$, (iii) each fork is of the form $f_j:=(p_j,\mathit{Up}_j,\loc_j)$ where (a) $p_j\in (0,1]$ is the fork probability, (b) $\mathit{Up}_j:\Rset^{n}\times \Rset^{m}\to \Rset^{n}$ is the update function of program valuations, and (c) $\loc_j$ is the destination location;
	\item $\mathcal{D}_0$ is the initial probability distribution over program valuations.
\end{itemize}

We assume that all PTSs satisfy the no demonic restriction. That is, for each location $\loc$, (1) the guards of any two different transitions starting from $\loc$ are \emph{mutually exclusive}: $\phi_i\wedge \phi_j\equiv \text{false}$ for $i\neq j$; 
(2) the $\phi_1,\dots,\phi_k$ guards of all transitions are \emph{mutually exhaustive}: $\bigvee_{i=1}^k \phi_i\equiv \text{true}$.

A \emph{state} of the PTS is a tuple $s=(\loc,\pv)$ where $\loc\in L$ is the location and $\pv\in\Rset^n$ is the program valuation. We denote the state space by $S$. A transition $\tau=(\loc,\phi,f_1,\dots,f_l)$ is \emph{enabled} at state $s=(\loc,\pv)$ if $\pv\models \phi$. A \emph{trace} of the PTS is an infinite seqence  of states, i.e, $\pi=(s_0,s_1,\dots)$. 
The semantics of the PTS is then given as follows: (1) it starts from the initial state $s_0=(\lin,\pv_0)$ with $\pv_0\sim \mathcal{D}_0$; (2) at the $n$-th step with the state $s_n=(\loc_n,\pv_n)$, if $\loc_n=\lout$, then $s_{n+1}=(\lout,\pv_n)$; otherwise, it picks a unique enabled transition $\tau=(\loc_n,\phi,f_1,\dots,f_u)$, and chooses a fork $f_j$ with probability $p_j(\rv_n)$, updates to $s_{n+1}=(\loc_{n+1},\pv_{n+1})$ where $\loc_{n+1}=\loc_j$ and $\pv_{n+1}=\mathit{Up}_j(\pv_n,\rv_n)$ with $\rv_n\sim\mathcal{D}$.

\noindent\textbf{Invariants.} Given a PTS $\Pi$, an \emph{invariant} is a function $I$ that assigns to each location $\loc\in L$, a predicate $I(\loc)$ over program variables. A state $(\loc,\pv)$ is \emph{reachable} if there exists a trace $\pi=(s_0,s_1,\dots)$ such that $s_n=(\loc,\pv)$ for some $n\in\Nset$. An invariant $I$ over-approximates the reachable states.

The execution of the PTS $\Pi$  gives rises to a probability space over the program traces via their probabilistic executions described above and standard constructions such as general state space Markov chains~\cite{meyn2012markov}. We denote the probability measure in this probability space by $\probm_{\valin}(-)$.

\subsection{Temporal Posterior Inference}
Consider a probabilistic program $P$ and its PTS $\Pi$. Let $AP=\{p_0,p_1,\dots,p_N\}$ be a finite set of atomic propositions that are predicates over program variables of the form $exp(\cdot)\ge 0$, where $exp:\Rset^n\rightarrow \Rset$.
Define a labeling function $\mathit{Lb}:S\rightarrow 2^{AP}$ such that for any trace $\pi=(s_0,s_1,\dots)$ of $\Pi$, there is a corresponding word $\sigma_{\pi}=(\mathit{Lb}(s_0),\mathit{Lb}(s_1),\dots)$ where each $\mathit{Lb}(s_t)\in 2^{AP}$. Let $\mathit{Words}(\Pi)$ be the set of all words induced by $\Pi$. 

Given an $\omega$-regular property $\varphi$ with its language $\mathcal{L}(\varphi)$, $\Pi\models \varphi$ iff $\mathit{Words}(\Pi)\subseteq \mathcal{L}(\varphi)$. Then we formalize the notion of temporal posterior distributions.
\begin{definition}[TPD]\label{def:tpd}
	The \emph{temporal posterior distribution} (TPD) $\tposterior_\Pi$ of the PTS $\Pi$ is defined by:
	\begin{align}\label{eq:TPD}
		\tposterior_{\Pi}(\varphi,\psi):=\frac{Z_\Pi(\varphi,\psi)}{Z_\Pi(\psi)}=\frac{\int_\calV \probm_{\valin}(\Pi\models\varphi \wedge  \psi ) \mu_{\mathrm{init}}(\mathrm{d} \valin)   }{\int_\calV \probm_{\valin}(\Pi\models\psi)\mu_{\mathrm{init}}(\mathrm{d} \valin) }
	\end{align}	
where $\varphi,\psi$ are two $\omega$-regular properties, $\mu_{\mathrm{init}}$ is the probability measure induced by the initial probability distribution $\mathcal{D}_0$, and $\calV:=\{ s\in S \mid s\in U \wedge  \mu_{\mathrm{init}}(U)>0  \}$ is the support of $\mu_{\mathrm{init}}$.	To be concise, we abuse the notations and write $\Pr(\varphi\mid\psi)$ to represent the TPD.	The TPD is \emph{integrable} if $0<Z_\Pi(\psi)<\infty$.
\end{definition}

\inlsec{Problem Statement} Given a PTS $\Pi$ and two $\omega$-regular properties $\varphi,\psi$, we want to find the upper and lower bounds $l,u\in [0,1]$ such that  $\tposterior_{\Pi}(\varphi,\psi)\in [l,u]$.

Note that SPI focuses on instantaneous properties at program termination. In fact, a program trace consists of a sequence of program states that describes the program's temporal behavior along with the time horizon. 
For example, in~\cref{ex:running-example}, $x=1$ keeps until program termination. 
We emphasize that our problem is non-trivial and more complex than SPI as it is difficult to quantitiavely track the infinite temporal behavior of a probabilistic program. 
According to~\cref{eq:TPD}, the problem can be reduced to find the quantitiave bounds on the satisfaction probabilities of $\omega$-regular properties, which will be discussed in the next section.

\section{Theoretical Results}\label{sec:theorem}

In this section, we present an automatic-theoretic method to make temporal reasoing about TPD feasible. Then we introduce new stochastic barrier certificates to establish quantitiave bounds on the satisfaction of $\omega$-regular properties.

\subsection{Automata-based Decomposition}

It is known that any $\omega$-regular property $\varphi$ over $AP$ can be translated into an DRA $\mathcal{A}_{\varphi}=(Q,q_0,2^{AP},\delta,\mathit{Acc})$ that accepts the same language of $\varphi$, i.e.,  $\mathcal{L}(\varphi)=\mathcal{L}(\mathcal{A}_{\varphi})$ (see~\cite{duret2022spot}).  

\inlsec{FOV and IOV-sets} Given a subset $U\subseteq Q$ of $\mathcal{A}_{\varphi}$, $U$ is \emph{finitely often visited} (FOV) and called an \emph{FOV-set} iff $\mathrm{Inf}(\rho)\cap U=\emptyset$ for any run $\rho\in\Xi$. $U$ is \emph{infinitely often visited} (IOV) and called an \emph{IOV-set} iff $\mathrm{Inf}(\rho)\cap U\neq\emptyset$ for any run $\rho\in\Xi$.

Consider a PTS $\Pi$ with a labeling function $\mathit{Lb}$ and an DRA $\mathcal{A}_{\varphi}$. We decompose the satisfaction of $\varphi$ via two types of product systems.

\inlsec{FOV and IOV products}  For a subset $U\subseteq Q$, an \emph{FOV product} of $\Pi$ and $\mathcal{A}_{\varphi}$ over the set $U$ is defined as a stochastic process $\Pi \otimes \mathcal{A}_{\varphi}^{U}=\{X_n\}_{n\in\Nset_0}$ that satisfies:
\begin{itemize}
	\item $X_0=(s_0,q_0,0)$, $s_0=(\lin,\pv_0)\in S_0$ with $\pv_0\sim\mathcal{D}_0$;
	\item $X_{n}=(s_n,q_n,l_n)$ for all $n\ge 0$, where $X_{n+1}=(s_{n+1},q_{n+1},l_{n+1})$ is the successor of $s_n$ such that $s_{n+1}$ is determined by the semantics of $\Pi$(see~\cref{sec:PPL}), $q_{n+1}=\delta(q_n,\mathit{Lb}(s_n))$, $l_{n+1}=l_n+1$ if $q_{n+1}\in U$ and $l_{n+1}=l_n$ otherwise.
	
\end{itemize}
 Likewise, an \emph{IOV product} of $\Pi$ and $\mathcal{A}_{\varphi}$ over the set $U$ is defined by a stochastic process $\Pi \otimes \mathcal{A}_{\varphi}^{U}=\{X_n\}_{n\in\Nset_0}$ that satisfies:
\begin{itemize}
	\item $X_0=(s_0,q_0,0)$, $s_0=(\lin,\pv_0)\in S_0$ with $\pv_0\sim\mathcal{D}_0$;
	\item $X_{n}=(s_n,q_n,l_n)$ for all $n\ge 0$, where $X_{n+1}=(s_{n+1},q_{n+1},l_{n+1})$ is  the successor of $s_n$ such that $s_{n+1}$ is determined by the semantcis of $\Pi$(see~\cref{sec:PPL}), $q_{n+1}=\delta(q_n,\mathit{Lb}(s_n))$, $l_{n+1}=0$ if $q_{n+1}\in U$ and $l_{n+1}=l_n+1$ otherwise.
\end{itemize}
The main difference between the two products is that we use a tick variable $l$ to count the number of entering $U$ in the FOV poduct and the number of  consecutively leaving $U$ in the IOV product, respectively.

A \emph{product trace} $\theta=(X_0,X_1,\dots)$ is an infinite sequence of product states $X_i$'s and the set of all product traces is denoted by $\Theta$. Let $\theta[2]=(q_0,q_1,\dots)$ be the corresponding run of the DRA that induced from $\theta$. Similarly, $\theta[3]=(h_0,h_1,\dots)$.
\begin{proposition}\label{prop:FOV}
Given an FOV product $\Pi \otimes \mathcal{A}_{\varphi}^{U}$, $U$ is an FOV-set, i.e., $\mathrm{Inf}(\theta[2])\cap U=\emptyset$ for any product trace $\theta\in\Theta$,  iff $\mathop{max}(\theta[3])<\infty$ for any $\theta\in\Theta$. 
\end{proposition}

\begin{proposition}\label{prop:IOV}
	Given an IOV product $\Pi \otimes \mathcal{A}_{\varphi}^{U}$, $U$ is an IOV-set, i.e., $\mathrm{Inf}(\theta[2])\cap U\neq\emptyset$ for any product trace $\theta\in\Theta$,  iff $\mathop{max}(\theta[3])<\infty$ for any $\theta\in\Theta$.
\end{proposition}

Recall~\cref{sec:omega-prelim},  the Rabin acceptance condition of $\mathcal{A}_{\varphi}$ is  $\mathit{Acc}=\{(E_i,F_i)\mid E_i,F_i\subseteq Q, E_i \cap F_i = \emptyset, 1\le i\le M\}$, and a run $\rho$ on a word $\sigma$ is accepting iff for some $i\in [1,M]$, $\mathrm{Inf}(\rho)\cap E_i=\emptyset$ and $\mathrm{Inf}(\rho)\cap F_i\neq\emptyset$. Note that the sets $E_i$'s shall be finitely often visited, which refers to persistence, while the sets $F_i$'s shall be infinitely often visited, which is related to recurrence.

\begin{corollary}\label{cor:FOV-IOV-omega}
If 	$\mathop{max}(\theta[3])<\infty$ for any $\theta\in\Theta$ in the FOV product $\Pi \otimes \mathcal{A}_{\varphi}^{E_i}$, then $E_i$ is an FOV-set. If $\mathop{max}(\theta[3])<\infty$ for any $\theta\in\Theta$ in the IOV product $\Pi \otimes \mathcal{A}_{\varphi}^{F_i}$, then $F_i$ is ian IOV-set. Finally, one can conclude that $\Pi\models \varphi$. 
\end{corollary}

\subsection{Quantitative Temporal Verification}\label{sec:quan-sto-NNCS}

In practice, it is not easy to determine whether the condition ``$\mathop{max}(\theta[3])<\infty$'' holds in~\cref{cor:FOV-IOV-omega}. To faciliate the calculation, below we present the variants of FOV and IOV-sets. 

\noindent
\textbf{$k$-times FOV and IOV-Sets.} Let $k\in\Nset_0$ be a non-negative constant. A set $U\subseteq Q$ is a \emph{$k$-times FOV-set} iff for any $\theta\in \Theta$, $U$ is visited at most $k$ times, i.e., $\mathop{max}(\theta[3])\le k$. Similarly, a set $U\subseteq Q$ is a \emph{$k$-times IOV-set} iff for all $\theta\in\Theta$, $Q\setminus U$ is visited consecutively at most $k$ times, i.e., $\mathop{max}(\theta[3])\le k$.

Fix a PTS $\Pi$ with a labeling function $\mathit{Lb}$ and an DRA $\mathcal{A}_\varphi$ with a set $U\subseteq Q$. 
Given a state $s=(\loc,\pv)$ of $\Pi$, its successor state $s'$ is computed as follows: (1) pick the unique enabled transition  at $s$,  $\tau=(\loc,\phi,f_1,\dots,f_{u})$, which has $u\ge 1$ forks of the form $f_j=(p_j,\mathit{Up}_j,\loc_j)$; (2) choose the $j$-th fork with probability $p_j(\rv)$ and then $s'=g(s,\rv)=\mathit{Up}_j(\pv,\rv)$ with $\rv\sim\mathcal{D}$ if $\loc\neq\lout$, otherwise $s'=g(s,\rv)=s$. Given a product state $(s,q)$ of $\Pi\otimes\mathcal{A}_{\varphi}^U$, its sucessor state is $(s',q')$ where $q'=\delta(q,\mathit{Lb}(s))$.

\begin{theorem}[Lower Bounds on $k$-times FOV-sets]\label{thm:lower-FOV}
Suppose there exists  a barrier certificate (BC) $\eta:S\times Q\times \Nset\rightarrow \Rset$ of the FOV product $\Pi \otimes \mathcal{A}_{\varphi}^{U}$ such that for two constants $\gamma\in (0,1)$ and $k\in\Nset_0$, the following conditions hold:
	\begin{align}
		& \eta(s,q,l)\ge 0 & \forall s\in S,q\in Q,l\in\Nset_0 \label{eq:SLPBC-nonnegative} \\
		& \eta(s,q_0,0)\le \gamma & \forall s\in S_0 \label{eq:SLPBC-initial} \\
		& \eta(s,q,k+1)\ge 1  & \forall s\in S, q\in U \label{eq:SLPBC-last} \\
		& \expv_{\rv\sim\mathcal{D}}[\eta(g(s,\rv),q',l') ] \le \eta(s,q,\ell)  & \forall s\in S, q\in Q, l\in [0,k] \label{eq:SLPBC-nonincrease}
	\end{align}
	where $l'=l+1$ if $q'\in U$ and $l$ otherwise.
	Then for any initial product state $(s_0,q_0,0)$ where $s_0\in S_0$, we have that $ \probm_{s_0} [ U \text{ is a } k \text{-times FOV-set}  ]\ge 1-\eta(s_0,q_0,0)=:l_i^{\mathit{fin}}$.

\end{theorem}
\noindent
\emph{Intuition.} Note that \cref{eq:SLPBC-nonincrease} is a supermartingale-type condition. This ensures that the BC is non-increasing
in expectation at each time step, which can be utilized to establish lower bounds on the satisfaction probability by Ville's Inequality~\cite{ville1939etude}.

\begin{theorem}[Upper Bounds on $k$-times FOV-sets]\label{thm:upper-FOV}
	Suppose there exists  a barrier certificate $\eta:S\times Q\times \Nset\rightarrow \Rset$ of the FOV product $\Pi \otimes \mathcal{A}_{\varphi}^{U}$ such that for some constants $\alpha\in (0,1)$, $0\le \lambda<\gamma\le 1$ and $k\in\Nset$, the following conditions hold:
\begin{align}
	& 0\le \eta(s,q,l)\le 1 & \forall s\in S,q\in Q,l\in\Nset_0 \label{eq:SUPBC-boundedness}\\
	& \eta(s,q_0,0)\ge \gamma & \forall s\in S_0 \label{eq:SUPBC-initial} \\
	& \eta(s,q,k+1)\le \lambda  & \forall s\in S, q\in U\label{eq:SUPBC-last} \\
	& \alpha\cdot \expv_{\rv\sim\mathcal{D}}[\eta(g(s,\rv),q',l') ] \ge \eta(s,q,\ell)  & \forall s\in S, q\in Q, l\in [0,k] \label{eq:SUPBC-increase}
\end{align}
	where $l'=l+1$ if $q'\in U$ and $l$ otherwise.
	Then for any initial product state $(s_0,q_0,0)$ where $s_0\in S_0$, we have that $\probm_{s_0} [U \text{ is a } k \text{-times FOV-set} ]\le 1-\eta(s_0,q_0,0)=:u_i^{\mathit{fin}}$.
	
\end{theorem}
\noindent
\emph{Intuition.}
\Cref{eq:SUPBC-increase} is a submartingale-type condition, which requires the expectation of the BC should increase after one time step. We use it to prove the upper bound on the satisfaction probability via the Optional Stopping Theorem (OST)~\cite{williams1991}.

\begin{theorem}[Lower Bounds on $k$-times IOV-sets]\label{thm:lower-IOV}
	Suppose there exists  a barrier certificate $h:S\times Q\times \Nset\rightarrow \Rset$ of the IOV product $\Pi \otimes \mathcal{A}_{\varphi}^{U}$ such that for two constants $\gamma\in (0,1)$ and $k\in\Nset$ , the following conditions hold:
\begin{align}
	& h(s,q,l)\ge 0 & \forall s\in S,q\in Q,l\in\Nset_0 \label{eq:SLRBC-nonnegative} \\
	& h(s,q_0,0)\le \gamma & \forall s\in S_0 \label{eq:SLRBC-initial} \\
	& h(s,q,k+1)\ge 1  & \forall s\in S, q\in Q\setminus U \label{eq:SLRBC-last} \\
	& \expv_{\rv\sim\mathcal{D}}[h(g(s,\rv),q',l') ] \le h(s,q,l)  & \forall s\in S, q\in Q, l\in [0,k] \label{eq:SLRBC-nonincrease}
\end{align}
	where $l'=0$ if $q'\in U$ and $l'=l+1$ otherwise.
	Then for any initial product state $(s_0,q_0,0)$ where $s_0\in S_0$, we have that $\probm_{s_0} [U \text{ is a } k \text{-times IOV-set}  ]\ge 1-h(s_0,q_0,0)=:l_i^{\mathit{inf}}$.
	
\end{theorem}

\begin{theorem}[Upper Bounds on $k$-times IOV-sets]\label{thm:upper-IOV}
	Suppose there exists  a barrier certificate $h:S\times Q\times \Nset\rightarrow \Rset$ of the IOV product $\Pi \otimes \mathcal{A}_{\varphi}^{U}$ such that for some constants $\gamma\in (0,1)$, $0\le \lambda<\gamma\le 1$ and $k\in\Nset$  , the following conditions hold:
\begin{align}
	& 0\le h(s,q,l)\le 1 & \forall s\in S,q\in Q,l\in\Nset_0 \label{eq:SURBC-boundedness} \\
	& h(s,q_0,0)\ge \gamma & \forall s\in S_0 \label{eq:SURBC-initial} \\
	& h(s,q,k+1)\le \lambda  & \forall s\in S, q\in Q\setminus U \label{eq:SURBC-last} \\
	& \alpha\cdot \expv_{\rv\sim\mathcal{D}}[h(g(s,\rv),q',l') ] \ge h(s,q,l)  & \forall s\in S, q\in Q, l\in [0,k] \label{eq:SURBC-increase}
\end{align}
	where $l'=0$ if $q'\in U$ and $l'=l+1$ otherwise.
	Then for any initial product state $(s_0,q_0,0)$ where $s_0\in S_0$, we have that $ \probm_{s_0} [U\text{ is a } k \text{-times IOV-set} ]\le 1-h(s_0,q_0,0)=:u_i^{\mathit{inf}}$.
	
\end{theorem}

Below we give the convergence theorem about the satisfaction of  $\omega$-regular properties under $k$-times FOV and IOV-sets. Let $(E_i,F_i)\in Acc$ be one element of the Rabin acceptance condition in $\mathcal{A}_\varphi$. $E_i$ is an	FOV-set such that $\mathop{max}(\theta[3])=m^*<\infty$ for any $\theta\in\Theta$, and $F_i$ is an IOV-set such that $\mathop{max}(\theta[3])=n^*<\infty$ for any $\theta\in\Theta$.

\begin{theorem}[Convergence]\label{thm:convergence}
For any initial state $s_0\in S_0$, when $k\rightarrow m^*$ in~\cref{thm:lower-FOV,thm:upper-FOV}, we have that $\probm_{s_0}[E_i \text{ is an FOV-set} ] \in [l_i^{\mathit{fin}},u_i^{\mathit{fin}}]$; when $k\rightarrow n^*$ in \cref{thm:lower-IOV,thm:upper-IOV}, we have that $\probm_{s_0}[F_i \text{ is an IOV-set} ] \in [l_i^{\mathit{inf}},u_i^{\mathit{inf}}]$. Hence, $\probm_{s_0}[\Pi\models\varphi]\in [ l_i^{\mathit{fin}}\cdot l_i^{\mathit{inf}},  u_i^{\mathit{fin}}\cdot u_i^{\mathit{inf}}]$.
\end{theorem}

\section{Algorithmic Approaches}\label{sec:method}
In this section, we provide template-based automated algorithms to synthesize the barrier certificates in~\cref{sec:theorem}. 

\noindent\textbf{Input and assumptions.} 
The input of our algorithm contains a PTS $\Pi$ with a labeling function $\mathit{Lb}$, an DRA $\mathcal{A}_{\varphi}$ that represents the $\omega$-regular property $\varphi$ (the states in $\mathcal{A}_{\varphi}$ can be natural numbers), an invariant $I$,\footnote{We assume that an invariant is provided as part of the input. Invariant generation is an orthogonal and well-studied problem, and can be automated using~\cite{solving2003linear,chatterjee2020polynomial}.} and techinical variables $d,k$ that specify the degree of polynomial templates and the counting number of barrier certificates (see, e.g., \cref{thm:lower-FOV}). We focus on polynomial PTSs, so we assume all guards and updates in $\Pi$ and all invariants $I(\loc)$'s are polynomial expressions over program variables.

\noindent\textbf{Output.} The task of our algorithm is to synthesize the polynomial barrier certificates $\eta,h$  in~\cref{sec:quan-sto-NNCS}, as well as the quantitative bounds from them.

\noindent\textbf{Overview.} Our algorithm is a standard template-based approach~\cite{chatterjee2016algorithmic,DBLP:conf/pldi/Wang0GCQS19,DBLP:journals/pacmpl/WangYFLO24}. We set up a polynomial template with unknown coefficients for each barrier certificate in~\cref{sec:quan-sto-NNCS}. Next, the conditions of each barrier certificate are encoded as  entailments of polynomial inequalities with unknown coefficients. Then by applying  the classical Farkas' Lemma~\cite{farkas1902theorie} or Putinar's Positivstellensatz~\cite{putinar1993positive}, we reduce the synthesis problem to Linear Programming (LP) or Semi-definite Programming (SDP). Finally, we solve the programming problem by calling an LP or SDP solver. Below we give the four steps of our algorithm. For brevity, we only showcase the BC in~\cref{thm:lower-FOV}, others can be derived in the same manner. 

\inlsec{Step 1 - Setting up templates}
 For each location $\loc\in L$,  our algorithm sets up a template $\eta(\loc)$ which is a polynomial consisting of all possible monomials of degree at most $d$ over program variables, each with an unknown coefficient. For example, in~\cref{ex:running-example}, the program has two program variables $x,y$ and two additional variables $q,l$ are required as in~\cref{thm:lower-FOV}. Let $d=1$, then at each location $\loc_i$, the algorithm sets $\eta(\loc_i)=c_{i,0}+c_{i,1}\cdot x+c_{i,2}\cdot y+c_{i,3}\cdot q+c_{i,4}\cdot l$ where $c_{i,j}$'s are unknown coefficents.

\inlsec{Step 2 - Constructing  entailments}
The algorithm symbolically computes the conditions of the barrier certificate in~\cref{thm:lower-FOV} using the templates from \textbf{Step 1}. Note that all of these conditions are entailments between polynomial inequalities over program variables whose coefficients are unknown. That is, they are in the form of ``$\forall \mathsf{x}.\ A(\mathsf{x})\Rightarrow b(\mathsf{x})$'', where $A$ is a set of polynomial inequalities over variables, and B is also a polynomial inequality over variables but with unknown coefficients. For example, for the program in~\cref{ex:running-example}, the algorithm symbolically computes \cref{eq:SLPBC-initial} where $q,l=0$ at the initial location, i.e.,
\begin{align}\label{eq:entail}
 \forall (x,y,q,l).\ (x\ge 0)\wedge (y\ge 0)\Rightarrow c_{0,0}+c_{0,1}\cdot x+c_{0,2}\cdot y\le \gamma.
\end{align}

\inlsec{Step 3 - Quantifier elimination}
After the previous step, we have a system of constraints of the form ``$\bigwedge_i (\forall \mathsf{x}.\ A_i(\mathsf{x})\rightarrow B_i(\mathsf{x})) $''.
The algorithm then eliminate the universal quantification over $\mathsf{x}$ in every constraint. We showcase the case where $A_i,B_i$ are linear inequalities. For example, the algorithm rewrites~\cref{eq:entail} as $\kappa_1\cdot x+\kappa_2\cdot y=\gamma-c_{0,0}-c_{0,1}\cdot x-c_{0,2}\cdot y$ with new unknown coefficients $\kappa_i\ge 0$. By Farkas' Lemma~\cite{farkas1902theorie}, we obtain a new system of linear equations soundly:
\begin{align*}
 0=\gamma-c_{0,0}, \qquad \kappa_1=-c_{0,1}, \qquad \kappa_2=-c_{0,2}.
\end{align*}
For polynomial cases, we refer the readers to~\cite{DBLP:journals/pacmpl/WangYFLO24}. 
 
\inlsec{Step 4 - Solving the system}
Since all the conditions of the BC are transformed into linear or semidefinite constraints, we call an LP or SDP solver to optimize the lower/upper bound from the BC and synthesize the unknown coefficients of the templates. For example, in~\cref{ex:running-example}, we compute the lower bound in~\cref{thm:lower-FOV}, and call an LP solver to maximize the objective function is $\eta(s_0,q_0,0)=c_{0,0}+c_{0,1}\cdot 1+c_{0,2}\cdot 0$ at the initial state $s_0=(\lin,1,0)$.

\begin{theorem}[Soundness]
If our algorithm finds valid solutions for the unknown coefficients in the polynomial templates, then it returns lower or upper bounds on the probabilities of $k$-times FOV-sets and IOV-sets.
\end{theorem}

\begin{proof}
 If the algorithm can synthesize the BC in~\cref{thm:lower-FOV}, then one can obtain the lower bound on the probability that $U$ is a $k$-time FOV-set. The same holds for \cref{thm:upper-FOV,thm:lower-IOV,thm:upper-IOV}.  
\end{proof}

\section{Evaluation}\label{sec:experiment}

We implement our approach into a prototype named \textsc{TPInfer}. The experiments aim to evaluate the effectiveness of our method with respect to several goals, including:
\begin{enumerate}
    \item Quantitative verification of $\omega$-regular properties.
    \item Quantitative verification of temporal posterior inference.
    \item Impact of hyperparameters $d,k$.
\end{enumerate}

\subsection{Benchmarks and Setup}
We evaluate the effectiveness of our method on five probabilistic programs, all of which contain loops with unbounded iterations. Specifically, ex3 and ex4 are taken from PSI~\cite{gehr2016psi}, 1d-asym-rwl is taken from \cite{zaiser2025guaranteed}, and the remaining two programs are novel examples proposed in this work. The $\omega$-regular properties we verify are distinguished between those that can be expressed in LTL and those that can only be expressed in automata (i.e., RE1 and RE2). All experiments are executed on a workstation running Ubuntu 22.04, with a 32-core AMD Ryzen Threadripper CPU, 128GB RAM, and a single 24564MiB GPU.

\subsection{Effectiveness of Quantitative Verification of $\omega$-regular properties}
We first verify the effectiveness of our method in synthesizing barrier certificates for Rabin persistence and recurrence acceptance, and in using them to compute sound upper and lower bounds on the probability that a probabilistic program satisfies a given $\omega$-regular property. The results are summarized in Table~\ref{tab:omega-regular-results}. For each $\omega$-regular specification, we compute the probabilities associated with its corresponding Rabin persistence and recurrence acceptance conditions, from which certified lower and upper bounds on the satisfaction probability are synthesized.

As shown in the table, the simulation-based probabilities (Sim.) consistently fall within the intervals defined by the synthesized lower (L.B.) and upper (U.B.) bounds. This empirically confirms the correctness and tightness of our verification results. The alignment between the simulation outcomes and the certified bounds demonstrates that the stochastic barrier certificates constructed by our method are able to capture the essential probabilistic behaviors of the system, even for properties that require reasoning over infinite executions.

Overall, the results validate the soundness of our approach and demonstrate its practical capability in quantitatively verifying a wide range of $\omega$-regular properties.

\begin{table}[t]
\centering
\footnotesize
\caption{$\omega$-regular Quantitative Verification Results}
\label{tab:omega-regular-results}
\resizebox{\textwidth}{!}{
\begin{tabular}{l l c c c c c c c}
\toprule
\textbf{Task} & \textbf{Property} & \textbf{Sim.} & \textbf{L.B.} & \textbf{U.B.} & \textbf{E.d.} & \textbf{F.d.}  & \textbf{k} & \textbf{time (sec)}\\
\midrule
\multirow{4}{*}{ex3}
& F(b=0)                     & 0.5       & 0.50000 & 0.75356 &1  &3 &3 & 1012.94\\ 
& F(b=0)\&F(n$\geq$2)            & 0.2493  & 0.17385 & 0.37294 &1  &3 &4 & 28944.35 \\
& GF(b=1)                    & 0.668     & 0.66670 & 0.87391 &1  &3 &3 & 1664.72\\
& GF(b=1)\&GF(n$\geq$2)          & 0.2493  & 0.18394 & 0.36288 &1  &3 &3 & 1881.14\\

\midrule
\multirow{4}{*}{ex4}
& F(c2=1)                    & 0.6665    & 0.61804 & 0.85172 &1  &3 &4 & 2395.88\\
& G(c1=0)\&F(c2=1)          & 0.3335   & 0.29054 & 0.53084 &3  &3 &4 & 6946.69\\
& GF(c2=1)                   & 0.6665    & 0.61903 & 0.89802 &1  &3 &7 & 14393.83\\
& GF(c1=0 \& c2=1)          & 0.3333 & 0.30193  & 0.49660 &1  &3 &5 & 7165.65  \\

\midrule
\multirow{4}{*}{1d-asym-rw}
& F(n$\geq$4)                     & 0.109     & 0.01962 & 0.10990 &1  &3 &3 & 1531.05\\
& G(x$\leq$5)\&F(n$\geq$4)            & 0.1068  & 0.09477 & 0.24967 &3  &3 &3 & 3140.63\\
& GF(x=0 $\rightarrow$ n$\geq$4) & 0.1094  & 0.05140 & 0.13178 &1  &3 &3 & 1399.08\\
& GF(n$\geq$3)                    & 0.2502   & 0.18586 & 0.31765 &1  &3 &4 & 1611.56\\

\midrule
\multirow{2}{*}{RE1}
& F(y$>$0)                      & 0.3563 & 0.27042 & 0.36236 &1  &3 &2 & 5512.71\\
& G(x=1)\&F(y$>$0)              & 0.1255  & 0.09462 & 0.19100 &3  &3 &2 & 15791.58 \\

\midrule
\multirow{2}{*}{RE2}
& (n$\leq$2)*(x=1)        & 0.8749   & 0.86578 & 0.94964 &3  &3 &2 &2620.16 \\
& (n$\leq$3)*(x=1)        & 0.9375  & 0.93582 & 0.97739 &3  &3 &3 & 2940.68\\
\bottomrule
\end{tabular}
}

\vspace{2mm}
\textbf{Remarks.}  
\textbf{Sim.}: probability via simulation.  
\textbf{L.B.}: certified lower bounds.  
\textbf{U.B.}: certified upper bounds.
\textbf{E.d.}: polynomial degree of the barrier certificate corresponding to Ei.
\textbf{F.d.}: polynomial degree of the barrier certificate corresponding to Fi.
\textbf{k}: the counting number of the barrier certificates.
\end{table}

\subsection{Effectiveness of Quantitative Verification of Temporal Posterior Inference}
Table~\ref{tab:omega-regular-conditional} demonstrates the effectiveness of our method in quantitatively verifying the temporal posterior probabilities of the $\omega$-regular properties under consideration. For each property, our approach computes certified lower and upper probability bounds through the synthesis of stochastic barrier certificates, and these bounds provide a sound enclosure of the true posterior probability.

As shown in~\cref{tab:omega-regular-conditional}, the simulation-based probabilities consistently fall within the synthesized intervals, indicating that the bounds we obtain are valid and informative. This agreement further confirms that the proposed barrier-certificate-based method is capable of accurately capturing the probabilistic behavior of programs with unbounded loops, even when reasoning about complex temporal properties defined in the $\omega$-regular framework. And the barrier certificate parameters used for computing the posterior probabilities are kept consistent with those employed in Table~\ref{tab:omega-regular-results}. The results highlight the reliability of our quantitative verification approach and its suitability for practical probabilistic analysis.

\begin{table}[t]
\centering
\footnotesize
\caption{Quantitative Verification Results of Posterior Probabilities}
\label{tab:omega-regular-conditional}
\resizebox{\textwidth}{!}{
\begin{tabular}{l l c c c}
\toprule
\textbf{Task} & \textbf{Property} & \textbf{Sim.} & \textbf{L.B.} & \textbf{U.B.} \\
\midrule

\multirow{2}{*}{ex3}
& F(n $\geq$ 2) $\mid$ F(b = 0)                          
& 0.5001 & 0.23071 & 0.74588 \\
& GF(n $\geq$ 2) $\mid$ GF(b = 1)              
& 0.2493 & 0.21048 & 0.54429 \\

\midrule
\multirow{2}{*}{ex4}
& G(c1 = 0) $\mid$ F(c2 = 1)                
& 0.5000 & 0.34113 & 0.85891 \\
& GF(c1 = 0\&c2 = 1) $\mid$ GF(c2 = 1)                  
& 0.4998 & 0.33622 & 0.80222 \\

\midrule
\multirow{2}{*}{1d-asym-rw}
& G(x $\leq$ 5)$\mid$ F(n $\geq$ 4)      
& 0.9743 & 0.86227 & 1 \\ 
& GF(x = 0 $\rightarrow$ n $\geq$ 4) $\mid$ GF(n $\geq$ 3)
& 0.4375 & 0.16181 & 0.70903 \\

\midrule
\multirow{1}{*}{RE1}
& G(x = 1) $\mid$ F(y $>$ 0)                               
& 0.4005 & 0.26112 & 0.70632 \\

\midrule
\multirow{1}{*}{RE2}
& (n $\leq$ 2)*(x = 1) $\mid$ (n $\leq$ 3)*(x = 1)  
& 0.9329 & 0.88581 & 1 \\

\bottomrule
\end{tabular}
}

\vspace{2mm}
\textbf{Remarks.}  
\textbf{Sim.}: probability via simulation.  
\textbf{L.B.}: certified lower bounds.  
\textbf{U.B.}: certified upper bounds.
\end{table}

\subsection{Impact of Polynomial Degree $d$ and the Counting Number $k$}
Table~\ref{tab:d-k-results} illustrates the impact of the polynomial degree $d$ and the counting number $k$ on the synthesis performance of the barrier certificates. When smaller polynomial degrees are used, the expressive power of the polynomial basis becomes limited, which may prevent the synthesis procedure from finding a feasible certificate and consequently lead to failure cases. This limitation reflects the fact that low-degree polynomials may not adequately capture the geometric structure of the invariants.

In particular, increasing $k$ helps refine the guaranties by allowing certificates to characterize longer execution behaviors, causing the computed probability bounds to gradually stabilize. The convergence of the lower and upper bounds toward the simulation-based probability demonstrates the alignment between the quantitative verification results (see also~\cref{thm:convergence}) and the empirical execution outcomes.

However, Table~\ref{tab:d-k-results} also reveals a trade-off: larger values of $d$ and $k$ lead to a substantial increase in synthesis time. This is expected, as higher-degree polynomials enlarge the search space of the optimization problem, while larger visitation bounds impose more complex constraints on the certificate synthesis.

It is also worth noting that, due to floating-point precision issues in the numerical optimization solver, the synthesized tight lower bounds may occasionally exceed the simulation-based probability estimates. Such discrepancies occur only in borderline cases and do not affect the soundness of the method, as the certified bounds remain mathematically valid within the solver's precision model.
\begin{table}[t]
\centering
\footnotesize
\caption{The Influence of Varying $k$ and $d$ on the Synthesis Performance}
\label{tab:d-k-results}
\resizebox{\textwidth}{!}{
\begin{tabular}{
>{\centering\arraybackslash}p{0.7cm}
>{\centering\arraybackslash}p{3cm}
>{\centering\arraybackslash}p{1.3cm}
>{\centering\arraybackslash}p{1.3cm}
>{\centering\arraybackslash}p{1.3cm}
>{\centering\arraybackslash}p{0.5cm}
>{\centering\arraybackslash}p{0.5cm}
c
}

\toprule
\textbf{Task} & \textbf{Property} & \textbf{Sim.} & \textbf{L.B.} & \textbf{U.B.} & \textbf{d} & \textbf{k} & \textbf{time} \\
\midrule

ex3 & F(b = 0) & 0   & 0         & 0.7996 & 3 & 1 & 197.0999 \\
    &          & 0.5 & 0.5002  & \textbf{fail} & 3 & 2 & 531.289 \\
    &          & 0.5 & 0.5000  & 0.7535 & 3 & 3 & 1012.9471 \\
    &          & 0.5 & 0.5077  & 0.8645 & 3 & 4 & 1662.5991 \\
    &          & 0.5 & 0.5020  & 0.8835 & 3 & 5 & 2426.8311 \\
\midrule

RE2 & (n $\leq$ 2)* (x = 1) & 0.8749 & 0.6551 & 0.9778 & 2 & 1 & 34.7432 \\
    &                       & 0.8749 & \textbf{fail}     & 0.9829 & 2 & 2 & 94.5126 \\
    &                       & 0.8749 & \textbf{fail}     & 0.9821 & 2 & 3 & 184.5404 \\
    &                       & 0.8749 & 0.8713 & 0.9675 & 3 & 1 & 233.3126 \\
    &                       & 0.8749 & 0.8657 & 0.9496 & 3 & 2 & 639.6580 \\
    &                       & 0.8749 & 0.8721 & 0.8901 & 3 & 3 & 1267.5365 \\
\bottomrule
\end{tabular}
}
\vspace{2mm}
\textbf{Remarks.}  
\textbf{Sim.}: probability via simulation.  
\textbf{L.B.}: certified lower bounds of F1.  
\textbf{U.B.}: certified upper bounds of F1.
\end{table}

\section{Related Work}\label{sec:related}
\noindent\textbf{Formal Verification of Probabilistic Programs.}
Formal verification of probabilistic programs has been widely studied through semantics, reasoning principles, and automated analysis frameworks. Foundational work provides weakest-preexpectation semantics and program logics for reasoning about probabilistic behavior~\cite{kaminski2016weakest}. Subsequent research develops automated techniques for proving almost-sure termination and safety using ranking supermartingales and expectation-based reasoning~\cite{chakarov2013probabilistic,chatterjee2016algorithmic,DBLP:conf/vmcai/FuC19}. These approaches focus on state-based properties at fixed points---typically at termination---and do not support posterior reasoning about temporal events along execution traces. Our work differs by introducing posterior inference for temporal properties conditioned on observations.

\noindent\textbf{Barrier Certificates for Stochastic Systems.}
Barrier certificates provide sufficient conditions for proving safety of continuous, hybrid, and stochastic systems. The original framework introduces barrier functions for excluding trajectories from unsafe sets~\cite{prajna2004safety}, later extended to stochastic dynamics using Lyapunov-like and sum-of-squares conditions~\cite{prajna2007framework,DBLP:journals/tac/0001FZ20}. While powerful for proving invariance and safety, these techniques do not quantify posterior probabilities or incorporate observational evidence. In contrast, our approach reasons directly about temporal posterior probabilities induced by probabilistic program traces.

\noindent\textbf{Temporal Logic Verification.}
Temporal logics such as LTL and $\omega$-regular specifications are standard tools for expressing safety and liveness properties over infinite behaviors~\cite{vardi1986automata}. For probabilistic systems, model checkers such as PRISM~\cite{kwiatkowska2011prism} and probabilistic temporal logics~\cite{baier2008principles} compute the prior probability that a temporal specification is satisfied. These methods do not incorporate Bayesian conditioning on observations or posterior reasoning about temporal events. Our work complements this line of research by enabling temporal posterior inference over the execution traces of probabilistic programs.

\section{Conclusion}\label{sec:conclusion}
We presented \emph{temporal posterior inference}, which extends classical posterior inference to entire execution traces of probabilistic programs. By combining program semantics with temporal logic, our method supports posterior reasoning about safety, liveness, and other $\omega$-regular properties under observational evidence. We introduced quantitative proof certificates, an automated synthesis algorithm, and demonstrated effectiveness on programs where standard inference fails. This provides a principled foundation for unifying Bayesian conditioning with temporal verification and invites further extensions to richer properties and larger models.

\clearpage
\bibliographystyle{splncs04}
 \bibliography{reference-TPI}

@incollection{thomas1990automata,
  title={Automata on infinite objects},
  author={Thomas, Wolfgang},
  booktitle={Handbook of Theoretical Computer Science, Volume B},
  editor={van Leeuwen, Jan},
  publisher={Elsevier},
  pages={133--191},
  year={1990}
}

@article{pnueli1977temporal,
  title={The temporal logic of programs},
  author={Pnueli, Amir},
  journal={18th Annual Symposium on Foundations of Computer Science (FOCS)},
  pages={46--57},
  year={1977}
}

@book{baier2008principles,
  title={Principles of Model Checking},
  author={Baier, Christel and Katoen, Joost-Pieter},
  publisher={MIT Press},
  year={2008}
}

@inproceedings{kwiatkowska2011prism,
  title={PRISM 4.0: Verification of probabilistic real-time systems},
  author={Kwiatkowska, Marta and Norman, Gethin and Parker, David},
  booktitle={International conference on computer aided verification},
  pages={585--591},
  year={2011},
  organization={Springer}
}

@incollection{gordon2014probabilistic,
  title={Probabilistic programming},
  author={Gordon, Andrew D and Henzinger, Thomas A and Nori, Aditya V and Rajamani, Sriram K},
  booktitle={Future of software engineering proceedings},
  pages={167--181},
  year={2014}
}

@article{goodman2012church,
  title={Church: a language for generative models},
  author={Goodman, Noah and Mansinghka, Vikash and Roy, Daniel M and Bonawitz, Keith and Tenenbaum, Joshua B},
  journal={arXiv preprint arXiv:1206.3255},
  year={2012}
}

@book{jaynes2003probability,
  title={Probability theory: The logic of science},
  author={Jaynes, Edwin T},
  year={2003},
  publisher={Cambridge university press}
}

@book{bishop2006pattern,
  title={Pattern recognition and machine learning},
  author={Bishop, Christopher M and Nasrabadi, Nasser M},
  volume={4},
  number={4},
  year={2006},
  publisher={Springer}
}

@article{ghahramani2015probabilistic,
  title={Probabilistic machine learning and artificial intelligence},
  author={Ghahramani, Zoubin},
  journal={Nature},
  volume={521},
  number={7553},
  pages={452--459},
  year={2015},
  publisher={Nature Publishing Group UK London}
}

@inproceedings{chakarov2013probabilistic,
  title={Probabilistic program analysis with martingales},
  author={Chakarov, Aleksandar and Sankaranarayanan, Sriram},
  booktitle={International Conference on Computer Aided Verification},
  pages={511--526},
  year={2013},
  organization={Springer}
}

@article{van2018introduction,
  title={An introduction to probabilistic programming},
  author={Van De Meent, Jan-Willem and Paige, Brooks and Yang, Hongseok and Wood, Frank},
  journal={arXiv preprint arXiv:1809.10756},
  year={2018}
}

@inproceedings{prajna2004safety,
  title={Safety verification of hybrid systems using barrier certificates},
  author={Prajna, Stephen and Jadbabaie, Ali},
  booktitle={International workshop on hybrid systems: Computation and control},
  pages={477--492},
  year={2004},
  organization={Springer}
}

@book{meyn2012markov,
  title={Markov chains and stochastic stability},
  author={Meyn, Sean P and Tweedie, Richard L},
  year={2012},
  publisher={Springer Science \& Business Media}
}

@inproceedings{duret2022spot,
  title={From spot 2.0 to spot 2.10: What’s new?},
  author={Duret-Lutz, Alexandre and Renault, Etienne and Colange, Maximilien and Renkin, Florian and Gbaguidi Aisse, Alexandre and Schlehuber-Caissier, Philipp and Medioni, Thomas and Martin, Antoine and Dubois, J{\'e}r{\^o}me and Gillard, Cl{\'e}ment and others},
  booktitle={International Conference on Computer Aided Verification},
  pages={174--187},
  year={2022},
  organization={Springer}
}

@inproceedings{chakarov2016deductive,
  title={Deductive proofs of almost sure persistence and recurrence properties},
  author={Chakarov, Aleksandar and Voronin, Yuen-Lam and Sankaranarayanan, Sriram},
  booktitle={International Conference on Tools and Algorithms for the Construction and Analysis of Systems},
  pages={260--279},
  year={2016},
  organization={Springer}
}

@inproceedings{gehr2016psi,
  title={PSI: Exact symbolic inference for probabilistic programs},
  author={Gehr, Timon and Misailovic, Sasa and Vechev, Martin},
  booktitle={International Conference on Computer Aided Verification},
  pages={62--83},
  year={2016},
  organization={Springer}
}

@article{zaiser2025guaranteed,
  title={Guaranteed Bounds on Posterior Distributions of Discrete Probabilistic Programs with Loops},
  author={Zaiser, Fabian and Murawski, Andrzej S and Ong, C-H Luke},
  journal={Proceedings of the ACM on Programming Languages},
  volume={9},
  number={POPL},
  pages={1104--1135},
  year={2025},
  publisher={ACM New York, NY, USA}
}

@article{ville1939etude,
  title={Etude critique de la notion de collectif},
  author={Ville, Jean},
  year={1939}
}

@book{williams1991,
  title={Probability with martingales},
  author={Williams, David},
  year={1991},
  publisher={Cambridge university press}
}

@article{solving2003linear,
  title={Linear Invariant Generation Using Non-linear},
  author={Solving, Constraint},
  journal={Computer-aided Verification: Proceedings},
  pages={420},
  year={2003},
  publisher={Springer-Verlag,.}
}

@inproceedings{chatterjee2020polynomial,
  title={Polynomial invariant generation for non-deterministic recursive programs},
  author={Chatterjee, Krishnendu and Fu, Hongfei and Goharshady, Amir Kafshdar and Goharshady, Ehsan Kafshdar},
  booktitle={Proceedings of the 41st ACM SIGPLAN Conference on Programming Language Design and Implementation},
  pages={672--687},
  year={2020}
}

@inproceedings{chatterjee2016algorithmic,
  title={Algorithmic analysis of qualitative and quantitative termination problems for affine probabilistic programs},
  author={Chatterjee, Krishnendu and Fu, Hongfei and Novotn{\`y}, Petr and Hasheminezhad, Rouzbeh},
  booktitle={Proceedings of the 43rd Annual ACM SIGPLAN-SIGACT Symposium on Principles of Programming Languages},
  pages={327--342},
  year={2016}
}

@inproceedings{DBLP:conf/pldi/Wang0GCQS19,
  author       = {Peixin Wang and
                  Hongfei Fu and
                  Amir Kafshdar Goharshady and
                  Krishnendu Chatterjee and
                  Xudong Qin and
                  Wenjun Shi},
  editor       = {Kathryn S. McKinley and
                  Kathleen Fisher},
  title        = {Cost analysis of nondeterministic probabilistic programs},
  booktitle    = {Proceedings of the 40th {ACM} {SIGPLAN} Conference on Programming
                  Language Design and Implementation, {PLDI} 2019, Phoenix, AZ, USA,
                  June 22-26, 2019},
  pages        = {204--220},
  publisher    = {{ACM}},
  year         = {2019},
  url          = {https://doi.org/10.1145/3314221.3314581},
  doi          = {10.1145/3314221.3314581},
  bibsource    = {dblp computer science bibliography, https://dblp.org}
}

@article{farkas1902theorie,
  title={Theorie der einfachen Ungleichungen.},
  author={Farkas, Julius},
  journal={Journal f{\"u}r die reine und angewandte Mathematik (Crelles Journal)},
  volume={1902},
  number={124},
  pages={1--27},
  year={1902},
  publisher={De Gruyter Berlin, New York}
}

@article{putinar1993positive,
  title={Positive polynomials on compact semi-algebraic sets},
  author={Putinar, Mihai},
  journal={Indiana University Mathematics Journal},
  volume={42},
  number={3},
  pages={969--984},
  year={1993},
  publisher={JSTOR}
}

@article{DBLP:journals/pacmpl/WangYFLO24,
  author       = {Peixin Wang and
                  Tengshun Yang and
                  Hongfei Fu and
                  Guanyan Li and
                  C.{-}H. Luke Ong},
  title        = {Static Posterior Inference of Bayesian Probabilistic Programming via
                  Polynomial Solving},
  journal      = {Proc. {ACM} Program. Lang.},
  volume       = {8},
  number       = {{PLDI}},
  pages        = {1361--1386},
  year         = {2024},
  url          = {https://doi.org/10.1145/3656432},
  doi          = {10.1145/3656432},
  bibsource    = {dblp computer science bibliography, https://dblp.org}
}

@article{feng2023lower,
  title={Lower bounds for possibly divergent probabilistic programs},
  author={Feng, Shenghua and Chen, Mingshuai and Su, Han and Kaminski, Benjamin Lucien and Katoen, Joost-Pieter and Zhan, Naijun},
  journal={Proceedings of the ACM on Programming Languages},
  volume={7},
  number={OOPSLA1},
  pages={696--726},
  year={2023},
  publisher={ACM New York, NY, USA}
}

@inproceedings{DBLP:conf/vmcai/FuC19,
  author    = {Hongfei Fu and
               Krishnendu Chatterjee},
  editor    = {Constantin Enea and
               Ruzica Piskac},
  title     = {Termination of Nondeterministic Probabilistic Programs},
  booktitle = {Verification, Model Checking, and Abstract Interpretation - 20th International
               Conference, {VMCAI} 2019, Cascais, Portugal, January 13-15, 2019,
               Proceedings},
  series    = {Lecture Notes in Computer Science},
  volume    = {11388},
  pages     = {468--490},
  publisher = {Springer},
  year      = {2019},
  doi       = {10.1007/978-3-030-11245-5\_22},
  bibsource = {dblp computer science bibliography, https://dblp.org}
}

@article{DBLP:journals/pacmpl/WangFCDX20,
  author    = {Peixin Wang and
               Hongfei Fu and
               Krishnendu Chatterjee and
               Yuxin Deng and
               Ming Xu},
  title     = {Proving expected sensitivity of probabilistic programs with randomized
               variable-dependent termination time},
  journal   = {Proc. {ACM} Program. Lang.},
  volume    = {4},
  number    = {{POPL}},
  pages     = {25:1--25:30},
  year      = {2020},
  doi       = {10.1145/3371093},
  bibsource = {dblp computer science bibliography, https://dblp.org}
}

@article{DBLP:journals/pacmpl/BartheEGHS18,
  author    = {Gilles Barthe and
               Thomas Espitau and
               Benjamin Gr{\'{e}}goire and
               Justin Hsu and
               Pierre{-}Yves Strub},
  title     = {Proving expected sensitivity of probabilistic programs},
  journal   = {Proc. {ACM} Program. Lang.},
  volume    = {2},
  number    = {{POPL}},
  pages     = {57:1--57:29},
  year      = {2018},
  doi       = {10.1145/3158145},
  bibsource = {dblp computer science bibliography, https://dblp.org}
}

@inproceedings{DBLP:conf/pldi/NgoC018,
  author    = {Van Chan Ngo and
               Quentin Carbonneaux and
               Jan Hoffmann},
  editor    = {Jeffrey S. Foster and
               Dan Grossman},
  title     = {Bounded expectations: resource analysis for probabilistic programs},
  booktitle = {Proceedings of the 39th {ACM} {SIGPLAN} Conference on Programming
               Language Design and Implementation, {PLDI} 2018, Philadelphia, PA,
               USA, June 18-22, 2018},
  pages     = {496--512},
  publisher = {{ACM}},
  year      = {2018},
  doi       = {10.1145/3192366.3192394},
  bibsource = {dblp computer science bibliography, https://dblp.org}
}

@inproceedings{DBLP:conf/tacas/BatzCJKKM23,
  author       = {Kevin Batz and
                  Mingshuai Chen and
                  Sebastian Junges and
                  Benjamin Lucien Kaminski and
                  Joost{-}Pieter Katoen and
                  Christoph Matheja},
  editor       = {Sriram Sankaranarayanan and
                  Natasha Sharygina},
  title        = {Probabilistic Program Verification via Inductive Synthesis of Inductive
                  Invariants},
  booktitle    = {Tools and Algorithms for the Construction and Analysis of Systems
                  - 29th International Conference, {TACAS} 2023, Held as Part of the
                  European Joint Conferences on Theory and Practice of Software, {ETAPS}
                  2022, Paris, France, April 22-27, 2023, Proceedings, Part {II}},
  series       = {Lecture Notes in Computer Science},
  volume       = {13994},
  pages        = {410--429},
  publisher    = {Springer},
  year         = {2023},
  url          = {https://doi.org/10.1007/978-3-031-30820-8\_25},
  doi          = {10.1007/978-3-031-30820-8\_25},
  bibsource    = {dblp computer science bibliography, https://dblp.org}
}

@inproceedings{abate2025quantitative,
  title={Quantitative supermartingale certificates},
  author={Abate, Alessandro and Giacobbe, Mirco and Roy, Diptarko},
  booktitle={International Conference on Computer Aided Verification},
  pages={3--28},
  year={2025},
  organization={Springer}
}

@inproceedings{henzinger2025supermartingale,
  title={Supermartingale certificates for quantitative omega-regular verification and control},
  author={Henzinger, Thomas A and Mallik, Kaushik and Sadeghi, Pouya and {\v{Z}}ikeli{\'c}, {\DJ}or{\dj}e},
  booktitle={International Conference on Computer Aided Verification},
  pages={29--55},
  year={2025},
  organization={Springer}
}

@article{watanabe2025unifying,
  title={A unifying approach to product constructions for quantitative temporal inference},
  author={Watanabe, Kazuki and Junges, Sebastian and Rot, Jurriaan and Hasuo, Ichiro},
  journal={Proceedings of the ACM on Programming Languages},
  volume={9},
  number={OOPSLA1},
  pages={1575--1603},
  year={2025},
  publisher={ACM New York, NY, USA}
}

@inproceedings{kaminski2016weakest,
  title={Weakest precondition reasoning for expected run--times of probabilistic programs},
  author={Kaminski, Benjamin Lucien and Katoen, Joost-Pieter and Matheja, Christoph and Olmedo, Federico},
  booktitle={European Symposium on Programming},
  pages={364--389},
  year={2016},
  organization={Springer}
}

@article{prajna2007framework,
  title={A Framework for Worst-Case and Stochastic Safety Verification Using Barrier Certificates},
  author={Prajna, Pablo and Jadbabaie, Ali and Pappas, George J.},
  journal={IEEE Transactions on Automatic Control},
  volume={52},
  number={8},
  pages={1415--1428},
  year={2007},
  publisher={IEEE},
  doi={10.1109/TAC.2007.901844}
}

@inproceedings{vardi1986automata,
  title={An Automata-Theoretic Approach to Automatic Program Verification},
  author={Vardi, Moshe Y. and Wolper, Pierre},
  booktitle={Proceedings of the Symposium on Logic in Computer Science (LICS)},
  pages={332--344},
  year={1986},
  publisher={IEEE},
  doi={10.1109/LICS.1986.1674816}
}

@article{DBLP:journals/tac/0001FZ20,
	author       = {Bai Xue and
	Martin Fr{\"{a}}nzle and
	Naijun Zhan},
	title        = {Inner-Approximating Reachable Sets for Polynomial Systems With Time-Varying
	Uncertainties},
	journal      = {{IEEE} Trans. Autom. Control.},
	volume       = {65},
	number       = {4},
	pages        = {1468--1483},
	year         = {2020}
}

\clearpage
\appendix
\begin{center}
    \Large \textbf{Appendix}
\end{center}

\section{Probability Theory}\label{app:probability-basis}

We start by reviewing some notions from probability theory.
\vskip 2pt 
\noindent\textbf{Random Variables and Stochastic Processes.}
A probability space is a triple ($\Omega, {\cal F}, {\mathbb P}$), where $\Omega$ is a non-empty sample space, ${\cal F}$ is a
$\sigma$-algebra over $\Omega$, and  $\mathbb P(\cdot)$ is a probability measure over $\cal F$, i.e. a function $\mathbb P$: ${\cal F}  \rightarrow [0,1]$ that satisfies the following properties: (1) ${\mathbb P}(\emptyset)=0$, (2)${\mathbb P}(\Omega - A)=1-{\mathbb P}[A]$, and (3) ${\mathbb P} (\textstyle\cup_{i=0}^{\infty} A_i)= \textstyle\sum_{i=0}^{\infty}{\mathbb P}(A_i)$ for  any
sequence $\{A_i\}_{i=0}^{\infty}$ of pairwise disjoint sets in ${\cal F}$.

Given a probability space ($\Omega, {\cal F}, {\mathbb P}$), a
random variable is a function $X: \Omega \rightarrow {\mathbb R} \cup \{\infty\}$ that is $\cal F$-measurable, i.e., for each $a \in {\mathbb R}$ we have that $\{\omega \in \Omega | X(\omega) \leq a \} \in {\cal F}$.
Moreover, a discrete-time stochastic process is a sequence $\{X_n\}_{n=0}^{\infty}$ of random variables in ($\Omega, {\cal F}, {\mathbb P}$).

\vskip 2pt 
\noindent\textbf{Conditional Expectation.}
Let ($\Omega, {\cal F}, {\mathbb P}$) be a probability space and $X$ be a random variable in ($\Omega, {\cal F}, {\mathbb P}$). The expected value of the random variable $X$,
denoted by ${\mathbb E}[X]$, is the Lebesgue integral of $X$ wrt $\mathbb P$. If the range
of $X$ is a countable set $A$, then ${\mathbb E}[X]= \textstyle\sum_{\omega \in A}\omega \cdot {\mathbb P}(X=\omega)$. 
Given a sub-sigma-algebra ${\cal F}' \subseteq {\cal F}$, a conditional expectation of $X$ for the given ${\cal F}'$ is a ${\cal F}'$-measurable random variable $Y$ such
that, for any $A \in {\cal F}'$, we have:
	\begin{equation}
	\begin{split}
     {\mathbb E}[X \cdot {\mathbb I}_{A}]={\mathbb E}[Y \cdot {\mathbb I}_{A}]
	\end{split}
	\end{equation}

\noindent Here, ${\mathbb I}_{A}: \Omega \rightarrow \{0,1\}$  is an indicator function of $A$, defined as ${\mathbb I}_{A}(\omega)=1$ if $\omega \in A$ and ${\mathbb I}_{A}(\omega)=0$ if $\omega \notin A$. 
Moreover, whenever the conditional expectation exists, it is also almost-surely unique, i.e., for any two ${\cal F}'$-measurable random variables $Y$ and $Y'$ which are conditional expectations of $X$ for given ${\cal F}'$, we have that ${\mathbb P}(Y=Y')=1$.

\vskip 2pt 
\noindent\textbf{Filtrations and Stopping Times.}
A filtration of the probability space ($\Omega, {\cal F}, {\mathbb P}$) is an infinite sequence $\{{{\cal F}_n}\}_{n=0}^{\infty}$ such that
for every $n$, the triple ($\Omega, {\cal F}_n, {\mathbb P}$) is a probability space and ${\cal F}_n \subseteq {\cal F}_{n+1} \subseteq {\cal F}$. 
A stopping time with respect to a filtration $\{{{\cal F}_n}\}_{n=0}^{\infty}$  is a random variable $T: \Omega \rightarrow {\mathbb N}_0 \cup \{\infty\}$ such that, for every $i \in {\mathbb N}_0$, it holds that $\{\omega \in \Omega |T(\omega) \leq i \} \in {\cal F}_i$. 
Intuitively, $T$ returns the time step at which some stochastic process shows a desired behavior and should be “stopped”.

A discrete-time stochastic process  $\{X_n\}_{n=0}^{\infty}$ in ($\Omega, {\cal F}, {\mathbb P}$) is adapted to a filtration $\{{{\cal F}_n}\}_{n=0}^{\infty}$, if
for all $n \geq 0$, $X_n$ is a random variable in ($\Omega, {\cal F}_n, {\mathbb P}$).

\vskip 2pt 
\noindent\textbf{Stopped Process.} Let $\{X_n\}_{n=0}^{\infty}$ be a stochastic process adapted to a filtration $\{\mathcal{F}_n\}_{n=0}^\infty$ and let $T$ be a stopping time w.r.t. $\{\mathcal{F}_n\}_{n=0}^\infty$. The stopped process $\{\tilde{X}_n\}_{n=0}^\infty$ is defined by
$$
\tilde{X}_n=
\begin{cases}
X_n, & \text{ for } n<T, \\
X_T, & \text{ for } n\ge T.
\end{cases}
$$

\vskip 2pt 
\noindent\textbf{Martingales.}
A discrete-time stochastic process $\{X_n\}_{n=0}^\infty$ to a filtration $\{{{\cal F}_n}\}_{n=0}^{\infty}$ is a martingale (resp. supermartingale, submartingale) if for all $n \geq 0$, ${\mathbb E}[|X_n|] < \infty$ and it holds almost surely (i.e., with probability 1) that ${\mathbb E}[X_{n+1}|{\cal F}_n] = X_n$ (resp. ${\mathbb E}[X_{n+1}|{\cal F}_n] \leq X_n$, ${\mathbb E}[X_{n+1}|{\cal F}_n] \geq X_n$).

\begin{theorem}[Optional Stopping Theorem (OST)~\cite{williams1991}]\label{thm:OST}
   Consider a stopping time $T$ w.r.t. a filtration $\{\mathcal{F}_n\}_{n=0}^{\infty}$ and a martingale (resp. supermartingale, submartingale) $\{X_n\}_{n=0}^{\infty}$ adapted to  $\{\mathcal{F}_n\}_{n=0}^{\infty}$. Then $\expv [|X_T|]<\infty$ and $\expv[X_T]=\expv [X_0]$ (resp. $\expv[X_T]\le \expv[X_0]$, $\expv[X_T]\ge \expv[X_0]$) if one of the following conditions holds:
   \begin{itemize}
       \item $T$ is bounded, i.e., $T<c$ for some constant $c$;
       \item $\expv[T]<\infty$, and there exists a constant $c$ such that for all $n\in \Nset$, $|X_{n+1}-X_n|\le c$;
       \item The stopped process $\{\tilde{X}_n\}_{n=0}^{\infty}$ w.r.t. $T$ is bounded, i.e., there exists some constant $c$ such that $|\tilde{X}_n|\le c$ for all $n\in\Nset$.
   \end{itemize}
\end{theorem}

\begin{theorem}[Ville's Inequality~\cite{ville1939etude}]\label{thm:ville}
For any nonnegative supermartingale $\{X_n\}_{n=0}^\infty$ and any real number $c>0$,
\[
\probm [\mathrm{sup}_{n\ge 0} X_n\ge c]\le \frac{\expv[X_0]}{c}.
\]  
\end{theorem}
This theorem is also called Doob's nonnegative supermartingale inequality.

\section{Supplementary Materials for \cref{sec:theorem}}\label{app:sec-sto-BCs}

\noindent\textbf{\cref{prop:FOV}}
Given an FOV product $\Pi \otimes \mathcal{A}_{\varphi}^{U}$, $U$ is an FOV-set, i.e., $\mathrm{Inf}(\theta[2])\cap U=\emptyset$ for any product trace $\theta\in\Theta$,  iff $\mathop{max}(\theta[3])<\infty$ for any $\theta\in\Theta$. 

\begin{proof}
  It is straightforward to see that in the FOV product, for any product trace $\theta\in\Theta$, $\mathop{max}(\theta[3])<\infty$ implies  $\mathrm{Inf}(\theta[2])\cap U=\emptyset$, and vice versa. \qed
\end{proof}

\noindent\textbf{\cref{prop:IOV}}
	Given an IOV product $\Pi \otimes \mathcal{A}_{\varphi}^{U}$, $U$ is an IOV-set, i.e., $\mathrm{Inf}(\theta[2])\cap U\neq\emptyset$ for any product trace $\theta\in\Theta$,  iff $\mathop{max}(\theta[3])<\infty$ for any $\theta\in\Theta$.

\begin{proof}
  It is straightforward to see that in the IOV product, for any product trace $\theta\in\Theta$, $\mathop{max}(\theta[3])<\infty$ implies  $\mathrm{Inf}(\theta[2])\cap U=\emptyset$, and vice versa. \qed
\end{proof}

\noindent\textbf{\cref{cor:FOV-IOV-omega}}
If 	$\mathop{max}(\theta[3])<\infty$ for any $\theta\in\Theta$ in the FOV product $\Pi \otimes \mathcal{A}_{\varphi}^{E_i}$, then $E_i$ is an FOV-set. If $\mathop{max}(\theta[3])<\infty$ for any $\theta\in\Theta$ in the IOV product $\Pi \otimes \mathcal{A}_{\varphi}^{F_i}$, then $F_i$ is ian IOV-set. Finally, one can conclude that $\Pi\models \varphi$. 

\begin{proof}
 By \cref{prop:FOV,prop:IOV}, one can prove that $E_i$ is an FOV-set and $F_i$ is an IOV-set. Therefore, according to the Rabin acceptance condition in~\cref{sec:omega-prelim}, one can conclude that $\Pi$ satisfies the $\omega$-regular property $\varphi$.   
\end{proof}

\noindent\textbf{\cref{thm:lower-FOV} (Lower Bounds on $k$-times FOV-sets).}
Suppose there exists  a barrier certificate (BC) $\eta:S\times Q\times \Nset\rightarrow \Rset$ of the FOV product $\Pi \otimes \mathcal{A}_{\varphi}^{U}$ such that for two constants $\gamma\in (0,1)$ and $k\in\Nset_0$, the following conditions hold:
	\begin{align}
		& \eta(s,q,l)\ge 0 & \forall s\in S,q\in Q,l\in\Nset_0 \tag{\ref{eq:SLPBC-nonnegative}} \\
		& \eta(s,q_0,0)\le \gamma & \forall s\in S_0 \tag{\ref{eq:SLPBC-initial}} \\
		& \eta(s,q,k+1)\ge 1  & \forall s\in S, q\in U \tag{\ref{eq:SLPBC-last}} \\
		& \expv_{\rv\sim\mathcal{D}}[\eta(g(s,\rv),q',l') ] \le \eta(s,q,\ell)  & \forall s\in S, q\in Q, l\in [0,k] \tag{\ref{eq:SLPBC-nonincrease}}
	\end{align}
	where $l'=l+1$ if $q'\in U$ and $l$ otherwise.
	Then for any initial product state $(s_0,q_0,0)$ where $s_0\in S_0$, we have that $ \probm_{s_0} [ U \text{ is a } k \text{-times FOV-set}  ]\ge 1-\eta(s_0,q_0,0)=:l_i^{\mathit{fin}}$.

\begin{proof}
Let $\Omega$ be the set of all trajectories induced by the FOV product $\Pi \otimes \mathcal{A}_{\varphi}^{U}$. Construct a stochastic process $\{X_t\}_{t\ge 0}$ over $\Omega$ such that $X_t(\omega)=\eta(\omega_t)=\eta(s_t,q_t,\ell_t)$ for any $\omega\in\Omega$. Let $\kappa$ be the first time that $U$ is visited $k+1$ times, i.e., $\kappa(\omega):=\mathrm{inf}\{t\in\Nset\mid \ell_t=k+1 \}$. Then we construct a stopped process $\{ \tilde{X}_t \}_{t\ge 0}$ by
$$
\tilde{X}_t=
\begin{cases}
X_t, & \text{ if } t<\kappa, \\
X_\kappa, & \text{ if } t\ge \kappa.
\end{cases}
$$
Next, we will prove that $\{ \tilde{X}_t \}_{t\ge 0}$ is a non-negative supermartingale. When $t<\kappa$, we have that 
\[
\expv[\tilde{X}_{t+1}\mid \tilde{\mathcal{F}}_t  ] =\expv[ X_{t+1}\mid \mathcal{F}_t  ]\le X_t=\tilde{X}_t,
\]
where the inequality is obtained by \cref{eq:SLPBC-nonincrease}. When $t\ge\kappa$, we have that
\[
\expv[\tilde{X}_{t+1}\mid \tilde{\mathcal{F}}_t  ]=X_\kappa=\tilde{X}_\kappa.
\]
By~\cref{eq:SLPBC-nonnegative}, we can conclude that $\{\tilde{X}_t\}_{t\ge 0}$ is a non-negative supermartingale.
According to~\cref{eq:SLPBC-last}, we have that $X_t(\omega)=\eta(s_t,q_t,\ell_t)\ge 1$ when $q_t\in U$ and $\ell_t= k+1$. Thus, for every initial state $(s_0,q_0,\ell_0)$, we have that
\begin{align*}
    & 1-\probm_{s_0} [U \text{ is visited at most } k \text{ times} ] \\
    =\ & \probm_{s_0} [U \text{ is visited more than } k \text{ times} ] \\
    =\ & \probm_{s_0} [\{ (s_t,q_t,\ell_t)   \}_{t=0}^{\infty} \text{ where } \exists t\in\Nset.\ q_t\in E_i \text{ and } \ell_t= k+1 ] \\
    \le\  & \probm_{s_0} \left [\mathop{\mathrm{sup}}\limits_{t\in\Nset} X_t(\omega) \ge 1  \right ] = \probm_{s_0} \left [\mathop{\mathrm{sup}}\limits_{t\in\Nset} \tilde{X}_t(\omega) \ge 1  \right ]\\
    \le\  & \tilde{X}_0(\omega) =X_0(\omega)
\end{align*}
 where the last inequality is obtained by Ville's Inequality(see~\cref{thm:ville}).   \qed
\end{proof}

\noindent\textbf{\cref{thm:upper-FOV} (Upper Bounds on $k$-times FOV-sets).}
	Suppose there exists  a barrier certificate $\eta:S\times Q\times \Nset\rightarrow \Rset$ of the FOV product $\Pi \otimes \mathcal{A}_{\varphi}^{U}$ such that for some constants $\alpha\in (0,1)$, $0\le \lambda<\gamma\le 1$ and $k\in\Nset$, the following conditions hold:
\begin{align}
	& 0\le \eta(s,q,l)\le 1 & \forall s\in S,q\in Q,l\in\Nset_0 \tag{\ref{eq:SUPBC-boundedness}}\\
	& \eta(s,q_0,0)\ge \gamma & \forall s\in S_0 \tag{\ref{eq:SUPBC-initial}} \\
	& \eta(s,q,k+1)\le \lambda  & \forall s\in S, q\in U   \tag{\ref{eq:SUPBC-last}} \\
	& \alpha\cdot \expv_{\rv\sim\mathcal{D}}[\eta(g(s,\rv),q',l') ] \ge \eta(s,q,\ell)  & \forall s\in S, q\in Q, l\in [0,k] \tag{\ref{eq:SUPBC-increase}}
\end{align}
	where $l'=l+1$ if $q'\in U$ and $l$ otherwise.
	Then for any initial product state $(s_0,q_0,0)$ where $s_0\in S_0$, we have that $\probm_{s_0} [U \text{ is a } k \text{-times FOV-set} ]\le 1-\eta(s_0,q_0,0)=:u_i^{\mathit{fin}}$.
	
\begin{proof}
Let $\Omega$ be the set of all trajectories induced by the FOV product $M_\mu \otimes \mathcal{A}_{\varphi}^{U}$. Construct a stochastic process $\{X_t\}_{t\ge 0}$ over $\Omega$ such that $X_t(\omega)=\eta(\omega_t)=\eta(s_t,q_t,\ell_t)$ for any $\omega\in\Omega$. Let $\kappa$ be the first time that $U$ is visited $k+1$ times, i.e., $\kappa(\omega):=\mathrm{inf}\{t\in\Nset\mid \ell_t=k+1 \}$. Then construct a new stochastic process $\{Y_t\}_{t\ge 0}$ such that $Y_t:=\alpha^t X_t$. Define its stopped process $\{ \tilde{Y}_t \}_{t\ge 0}$ by
$$
\tilde{Y}_t=
\begin{cases}
Y_t, & \text{ if } t<\kappa, \\
Y_\kappa, & \text{ if } t\ge \kappa.
\end{cases}
$$
Next, we will prove that $\{ \tilde{Y}_t \}_{t\ge 0}$ a submartingale. When $t<\kappa$, we have that
\begin{align*}
   \expv[\tilde{Y}_{t+1}\mid \tilde{\mathcal{F}}_t] &= \expv[Y_{t+1}\mid \mathcal{F}_t] \\
   &= \expv[\alpha^{t+1}X_{t+1}\mid \mathcal{F}_t]=\alpha^{t+1} \expv [X_{t+1}\mid \mathcal{F}_t] \\
   &= \alpha^t\cdot \alpha\expv[X_{t+1}\mid \mathcal{F}_t] \ge \alpha^t X_t=\tilde{Y}_t
\end{align*}
where the inequality is obtained by~\cref{eq:SUPBC-increase}. When $t\ge \kappa$, we have that
\begin{align*}
    \expv[\tilde{Y}_{t+1}\mid \tilde{\mathcal{F}}_t]=Y_\kappa= \tilde{Y}_t
\end{align*}
By~\cref{eq:SUPBC-boundedness}, we have $\tilde{Y}_t\in [0,1]$ for any $t\in\Nset$. Thus, we can conclude that $\{ \tilde{Y}_t \}_{t\ge 0}$ is a submartingale. Then for every initial state $(s_0,q_0,\ell_0)$, we have that
\begin{align*}
 &  1-\probm_{s_0} [U \text{ is visited at most } k \text{ times} ] \\
 =\ & \probm_{s_0}[U \text{ is visited more than } k \text{ times} ] \\  
  =\ & \probm_{s_0} [\{ (s_t,q_t,\ell_t)   \}_{t=0}^{\infty} \text{ where } \exists t\in\Nset.\ q_t\in E_i \text{ and } \ell_t= k+1 ] \\
  = \ & \probm_{s_0} [\kappa <\infty] \\
 \ge \ & \expv[\alpha^\kappa\cdot X_{\kappa} \mid \kappa<\infty]\cdot \probm_{s_0}[\kappa<\infty] \\
 = \ & \expv[\alpha^\kappa \cdot X_{\kappa} \mid \kappa<\infty]\cdot \probm_{s_0}[\kappa<\infty]+\expv[\alpha^\kappa \cdot X_{\kappa}\mid \kappa=\infty]\cdot \probm_{s_0}[\kappa=\infty] \\
 =\  & \expv[\alpha^\kappa \cdot X_\kappa] = \expv[Y_\kappa] =\expv[\tilde{Y}_\kappa]   \ge \expv[\tilde{Y}_0]=X_0\\
\end{align*}
where the first inequality is derived from the fact that $\alpha^\kappa\cdot X_\kappa\in (0,1)$ as $\alpha\in (0,1)$ and $X_\kappa\in [0,1]$ by~\cref{eq:SUPBC-boundedness}, the fourth equality is obtained by the fact that $\alpha^\kappa \cdot X_\kappa=0$ when $\kappa=\infty$ and $\alpha\in (0,1)$, and the last inequality is got by applying the Optional Stopping Theorem~\cref{thm:OST}.  \qed
\end{proof}

\noindent\textbf{\cref{thm:lower-IOV} (Lower Bounds on $k$-times IOV-sets).}
	Suppose there exists  a barrier certificate $h:S\times Q\times \Nset\rightarrow \Rset$ of the IOV product $\Pi \otimes \mathcal{A}_{\varphi}^{U}$ such that for two constants $\gamma\in (0,1)$ and $k\in\Nset$ , the following conditions hold:
\begin{align}
	& h(s,q,l)\ge 0 & \forall s\in S,q\in Q,l\in\Nset_0 \tag{\ref{eq:SLRBC-nonnegative}} \\
	& h(s,q_0,0)\le \gamma & \forall s\in S_0 \tag{\ref{eq:SLRBC-initial}} \\
	& h(s,q,k+1)\ge 1  & \forall s\in S, q\in Q\setminus U \tag{\ref{eq:SLRBC-last}} \\
	& \expv_{\rv\sim\mathcal{D}}[h(g(s,\rv),q',l') ] \le h(s,q,l)  & \forall s\in S, q\in Q, l\in [0,k] \tag{\ref{eq:SLRBC-nonincrease}}
\end{align}
	where $l'=0$ if $q'\in U$ and $l'=l+1$ otherwise.
	Then for any initial product state $(s_0,q_0,0)$ where $s_0\in S_0$, we have that $\probm_{s_0} [U \text{ is a } k \text{-times IOV-set}  ]\ge 1-h(s_0,q_0,0)=:l_i^{\mathit{inf}}$.

\begin{proof}
Let $\Omega$ be the set of all trajectories induced by the IOV product $M_\mu \otimes \mathcal{A}_{\varphi}^{U}$. Construct a stochastic process $\{X_t\}_{t\ge 0}$ over $\Omega$ such that $X_t(\omega)=h(\omega_t)=h(s_t,q_t,\ell_t)$ for any $\omega\in\Omega$. 
Let $\kappa$ be the first time that $Q\setminus U$ is visited consecutively $k+1$ times, i.e., $\kappa(\omega):=\mathrm{inf}\{t\in\Nset\mid \ell_t=k+1 \}$. 
Then define a stopped process $\{ \tilde{X}_t \}_{t\ge 0}$ by
$$
\tilde{X}_t=
\begin{cases}
X_t, & \text{ if } t<\kappa, \\
X_\kappa, & \text{ if } t\ge \kappa.
\end{cases}
$$
Next, we will prove that $\{ \tilde{X}_t \}_{t\ge 0}$ is a non-negative supermartingale. When $t<\kappa$, we have that 
\[
\expv[\tilde{X}_{t+1}\mid \tilde{\mathcal{F}}_t  ] =\expv[ X_{t+1}\mid \mathcal{F}_t  ]\le X_t=\tilde{X}_t,
\]
where the inequality is obtained by \cref{eq:SLRBC-nonincrease}. When $t\ge\kappa$, we have that
\[
\expv[\tilde{X}_{t+1}\mid \tilde{\mathcal{F}}_t  ]=X_\kappa=\tilde{X}_\kappa.
\]
By~\cref{eq:SLRBC-nonnegative}, we can conclude that $\{ \tilde{X}_t \}_{t\ge 0}$ is a non-negative supermartingale. According to~\cref{eq:SLRBC-last}, we have that $X_t(\omega)=h(s_t,q_t,\ell_t)\ge 1$ when $q_t\in Q\setminus F_i$ and $\ell_t=k+1$. Thus, for every initial state $(s_0,q_0,\ell_0)$, we have that
\begin{align*}
    & 1-\probm_{s_0} [Q\setminus U \text{ is visited consecutively at most } k \text{ times} ] \\
    =\ & \probm_{s_0} [Q\setminus U \text{ is visited consecutively more than } k \text{ times} ] \\
    =\ & \probm_{s_0} [\{ (s_t,q_t,\ell_t)   \}_{t=0}^{\infty} \text{ where } \exists t\in\Nset.\ q_t\in Q\setminus F_i \text{ and } \ell_t= k+1 ] \\
    \le\  & \probm_{s_0} \left [\mathop{\mathrm{sup}}\limits_{t\in\Nset} X_t(\omega) \ge 1  \right ] = \probm_{s_0} \left [\mathop{\mathrm{sup}}\limits_{t\in\Nset} \tilde{X}_t(\omega) \ge 1  \right ]\\
    \le\  & \tilde{X}_0(\omega) =X_0(\omega)
\end{align*}
 where the last inequality is obtained by Ville's Inequality~\cref{thm:ville}.   \qed

\end{proof}

\noindent\textbf{\cref{thm:upper-IOV} (Upper Bounds on $k$-times IOV-sets).}
	Suppose there exists  a barrier certificate $h:S\times Q\times \Nset\rightarrow \Rset$ of the IOV product $\Pi \otimes \mathcal{A}_{\varphi}^{U}$ such that for some constants $\gamma\in (0,1)$, $0\le \lambda<\gamma\le 1$ and $k\in\Nset$  , the following conditions hold:
\begin{align}
	& 0\le h(s,q,l)\le 1 & \forall s\in S,q\in Q,l\in\Nset_0 \tag{\ref{eq:SURBC-boundedness}} \\
	& h(s,q_0,0)\ge \gamma & \forall s\in S_0 \tag{\ref{eq:SURBC-initial}} \\
	& h(s,q,k+1)\le \lambda  & \forall s\in S, q\in Q\setminus U \tag{\ref{eq:SURBC-last}} \\
	& \alpha\cdot \expv_{\rv\sim\mathcal{D}}[h(g(s,\rv),q',l') ] \ge h(s,q,l)  & \forall s\in S, q\in Q, l\in [0,k] \tag{\ref{eq:SURBC-increase}}
\end{align}
	where $l'=0$ if $q'\in U$ and $l'=l+1$ otherwise.
	Then for any initial product state $(s_0,q_0,0)$ where $s_0\in S_0$, we have that $ \probm_{s_0} [U\text{ is a } k \text{-times IOV-set} ]\le 1-h(s_0,q_0,0)=:u_i^{\mathit{inf}}$.

\begin{proof}
Let $\Omega$ be the set of all trajectories induced by the IOV product $M_\mu \otimes \mathcal{A}_{\varphi}^{U}$. Construct a stochastic process $\{X_t\}_{t\ge 0}$ over $\Omega$ such that $X_t(\omega)=h(\omega_t)=h(s_t,q_t,\ell_t)$ for any $\omega\in\Omega$. Let $\kappa$ be the first time that $Q\setminus U$ is visited consecutively $k+1$ times, i.e., $\kappa(\omega):=\mathrm{inf}\{t\in\Nset\mid \ell_t=k+1 \}$. Then construct a new stochastic process $\{Y_t\}_{t\ge 0}$ such that $Y_t:=\alpha^t X_t$. Define its stopped process $\{ \tilde{Y}_t \}_{t\ge 0}$ by
$$
\tilde{Y}_t=
\begin{cases}
Y_t, & \text{ if } t<\kappa, \\
Y_\kappa, & \text{ if } t\ge \kappa.
\end{cases}
$$
Next, we will prove that $\{ \tilde{Y}_t \}_{t\ge 0}$ a submartingale. When $t<\kappa$, we have that
\begin{align*}
   \expv[\tilde{Y}_{t+1}\mid \tilde{\mathcal{F}}_t] &= \expv[Y_{t+1}\mid \mathcal{F}_t] \\
   &= \expv[\alpha^{t+1}X_{t+1}\mid \mathcal{F}_t]=\alpha^{t+1} \expv [X_{t+1}\mid \mathcal{F}_t] \\
   &= \alpha^t\cdot \alpha\expv[X_{t+1}\mid \mathcal{F}_t] \ge \alpha^t X_t=\tilde{Y}_t
\end{align*}
where the inequality is obtained by~\cref{eq:SURBC-increase}. When $t\ge \kappa$, we have that
\begin{align*}
    \expv[\tilde{Y}_{t+1}\mid \tilde{\mathcal{F}}_t]=Y_\kappa= \tilde{Y}_t
\end{align*}
By~\cref{eq:SURBC-boundedness}, we have $\tilde{Y}_t\in [0,1]$ for any $t\in\Nset$. Thus, we can conclude that $\{ \tilde{Y}_t \}_{t\ge 0}$ is a submartingale. Then for every initial state $(s_0,q_0,\ell_0)$, we have that
\begin{align*}
 &  1-\probm_{s_0} [Q\setminus U \text{ is visited consecutively at most } k \text{ times} ] \\
 =\ & \probm_{s_0}[Q\setminus U \text{ is visited consecutively more than } k \text{ times} ] \\  
  =\ & \probm_{s_0} [\{ (s_t,q_t,\ell_t)   \}_{t=0}^{\infty} \text{ where } \exists t\in\Nset.\ q_t\in Q\setminus F_i \text{ and } \ell_t= k+1 ] \\
  = \ & \probm_{s_0} [\kappa <\infty] \\
 \ge \ & \expv[\alpha^\kappa\cdot X_{\kappa} \mid \kappa<\infty]\cdot \probm_{s_0}[\kappa<\infty] \\
 = \ & \expv[\alpha^\kappa \cdot X_{\kappa} \mid \kappa<\infty]\cdot \probm_{s_0}[\kappa<\infty]+\expv[\alpha^\kappa \cdot X_{\kappa}\mid \kappa=\infty]\cdot \probm_{s_0}[\kappa=\infty] \\
 =\  & \expv[\alpha^\kappa \cdot X_\kappa] = \expv[Y_\kappa] =\expv[\tilde{Y}_\kappa]   \ge \expv[\tilde{Y}_0]=X_0\\
\end{align*}
where the first inequality is derived from the fact that $\alpha^\kappa\cdot X_\kappa\in (0,1)$ as $\alpha\in (0,1)$ and $X_\kappa\in [0,1]$ by~\cref{eq:SURBC-boundedness}, the fourth equality is obtained by the fact that $\alpha^\kappa \cdot X_\kappa=0$ when $\kappa=\infty$ and $\alpha\in (0,1)$, and the last inequality is got by applying the Optional Stopping Theorem~\cref{thm:OST}.  \qed
    
\end{proof}

\noindent\textbf{\cref{thm:convergence} (Convergence).}
For any initial state $s_0\in S_0$, when $k\rightarrow m^*$ in~\cref{thm:lower-FOV,thm:upper-FOV}, we have that $\probm_{s_0}[E_i \text{ is an FOV-set} ] \in [l_i^{\mathit{fin}},u_i^{\mathit{fin}}]$; when $k\rightarrow n^*$ in \cref{thm:lower-IOV,thm:upper-IOV}, we have that $\probm_{s_0}[F_i \text{ is an IOV-set} ] \in [l_i^{\mathit{inf}},u_i^{\mathit{inf}}]$. Hence, $\probm_{s_0}[\Pi\models\varphi]\in [ l_i^{\mathit{fin}}\cdot l_i^{\mathit{inf}},  u_i^{\mathit{fin}}\cdot u_i^{\mathit{inf}}]$.

\begin{proof}
   The proof is straightforward. When $k$ is sufficiently large enough, we will count the contributions of all valid traces. 
\end{proof}

\section{Supplementary Materials for \cref{sec:experiment}}

\begin{figure}[htbp]
\centering

\subfloat[RE 1]{
\begin{minipage}[t]{0.45\textwidth}
\raggedright
\hspace{10mm}\texttt{$x \sim \mathrm{Bernoulli}(0.3)$; $y := 0$; \\
\hspace{10mm}\textbf{while} $\mathrm{flip}(0.5)$ \textbf{do} \\
\hspace{10mm}\quad \textbf{if} $\mathrm{flip}(0.6)$ \textbf{then} \\
\hspace{10mm}\quad\quad $x := 1 - x$; \\
\hspace{10mm}\quad \textbf{if} $x = 1$ \textbf{then} \\
\hspace{10mm}\quad\quad $y := y + 1$; \\
\hspace{10mm}\textbf{od}
}
\end{minipage}
}
\hfill
\subfloat[RE 2]{
\begin{minipage}[t]{0.45\textwidth}
\raggedright
\hspace{10mm}\texttt{$x := 0$; $y := 0$; \\
\hspace{10mm}\textbf{while} $x = 0$ \textbf{do} \\
\hspace{10mm}\quad $n := n + 1$; \\
\hspace{10mm}\quad \textbf{if} $\mathrm{flip}(0.5)$ \textbf{then} \\
\hspace{10mm}\quad\quad $x := 1 - x$; \\
\hspace{10mm}\textbf{od}
}
\end{minipage}
}

\caption{The novel examples}
\label{fig:running-example1}
\end{figure} 
 
\end{document}